\documentclass[11pt,letterpaper]{article}
\usepackage[numbers]{natbib}
\usepackage{amsthm,amsmath,amssymb,amsfonts} 
\usepackage{epsfig} 
\usepackage{latexsym,nicefrac,bbm}
\usepackage{xspace}
\usepackage{color,fancybox,graphicx,subfigure,fullpage}
\usepackage[top=1.25in, bottom=1.25in, left=1in, right=1in]{geometry}
\usepackage{tabularx}
\usepackage{hyperref, color}
\hypersetup{colorlinks=true,citecolor=blue, linkcolor=blue, urlcolor=blue} 
\usepackage{pdfsync}
\usepackage{multicol}
\usepackage{cleveref}
\usepackage{stmaryrd}

\usepackage{float}

\usepackage{array}
\usepackage{makecell}

\renewcommand{\epsilon}{\varepsilon}

\newtheorem{theorem}{Theorem}[section]
\newtheorem{lemma}[theorem]{Lemma}
\newtheorem{definition}[theorem]{Definition}
\newtheorem*{theorem*}{Theorem}
\newtheorem{remark}[theorem]{Remark}

\usepackage{verbatim}

\newtheorem{problem}{Problem}
\newtheorem{assumption}{Assumption}

\usepackage{algpseudocode}
\usepackage{algorithm,caption}

\usepackage{mhequ}
\def \be{\begin{equs}}
\def \ee{\end{equs}}

\title{\bf Nonconvex sampling with the Metropolis-adjusted Langevin algorithm}
\author{Oren Mangoubi\thanks{{\'{E}cole Polytechnique F\'{e}d\'{e}rale de Lausanne (EPFL).}}
\and
Nisheeth K. Vishnoi\thanks{Yale University}
}
\date{}

\begin{document}
\maketitle

\begin{abstract}

The Langevin Markov chain algorithms are widely deployed methods to sample from distributions in challenging high-dimensional and non-convex statistics and machine learning applications.
Despite this, current bounds for the Langevin algorithms are slower than those of competing algorithms in many important situations, for instance when sampling from weakly log-concave distributions, or when sampling or optimizing non-convex log-densities.
  In this paper, we obtain improved bounds in many of these situations, showing that the Metropolis-adjusted Langevin algorithm (MALA) is faster than the best bounds for its competitor algorithms when the target distribution satisfies weak third- and fourth- order regularity properties associated with the input data.  
    In many settings, our regularity conditions are weaker than the usual Euclidean operator norm regularity properties, allowing us to show faster bounds for a much larger class of distributions than would be possible with the usual Euclidean operator norm approach,  including in statistics and machine learning applications where the data satisfy a certain incoherence condition.   
    In particular, we show that using our regularity conditions one can obtain faster bounds for applications which include sampling problems in Bayesian logistic regression with weakly convex priors, and the nonconvex optimization problem of learning linear classifiers with zero-one loss functions.

Our main technical contribution in this paper is our analysis of the Metropolis acceptance probability of MALA in terms of its ``energy-conservation error," and our bound for this error in terms of third- and fourth- order regularity conditions.
  Our combination of this higher-order analysis of the energy conservation error with the conductance method is key to obtaining bounds which have a sub-linear dependence on the dimension $d$ in the non-strongly logconcave setting.
\end{abstract}

\newpage
\tableofcontents
\newpage

\section{Introduction}
Sampling from a probability distribution is a fundamental algorithmic problem that arises in several areas including machine learning, statistics, optimization, theoretical computer science, and molecular dynamics.  %
In many situations, for instance when the dimension $d$ is large or the target distribution is nonconvex, sampling problems become computationally difficult, and MCMC algorithms are among the most popular methods used to solve them.

Formally, we consider the problem of sampling from a distribution $\pi(x) \propto e^{-U(x)}$, where one is given access to a function $U: \mathbb{R}^d \rightarrow \mathbb{R}$ and its gradient ${\nabla U}$:
\begin{problem}
Given access to a function $U: \mathbb{R}^d \rightarrow \mathbb{R}$ and its gradient ${\nabla U}$, an initial point $X_0$, and $\epsilon>0$, generate a sample with total variation error $\epsilon$ from the distribution $\pi(x) \propto e^{-U(x)}$.
\end{problem}

\noindent
We also consider the problem of optimizing a function $U$.  Any generic sampling method can also be used as an optimization technique: if one samples from the distribution  $\propto e^{- \mathcal{T}^{-1}U(x)}$ for a low enough temperature parameter $\mathcal{T}$ then the samples will concentrate near the global optima.  Specifically, we consider the problem of optimizing a function $U(x)$ on $\mathsf{S} \subseteq \mathbb{R}^d$, where one is given access to a function $U: \mathbb{R}^d \rightarrow \mathbb{R}$, its gradient ${\nabla U}$, and a membership oracle for $\mathsf{S}$:
\begin{problem}
Given access to a function $U: \mathbb{R}^d \rightarrow \mathbb{R}$ and its gradient ${\nabla U}$, a membership oracle for $\mathsf{S}$, an initial point $X_0$, and $\epsilon>0$, generate an approximate minimizer $\hat{x}^\star$ such that $F(\hat{x}^\star) - \inf_{x\in \mathsf{S}}F(x) \leq \epsilon$.
\end{problem}

  \noindent
  The Langevin Monte Carlo algorithms can be thought of as discretizations of the Langevin diffusion with invariant measure $\pi$.  
  The Langevin algorithms without Metropolis adjustment work by approximating a particular outcome of this diffusion. 
  For instance, each step of the unadjusted Langevin algorithm (ULA) Markov chain $\tilde{X}$ is given as  $\tilde{X}_{i+1} = \tilde{X}_i + \eta V_i- \frac{1}{2}\eta^2 \nabla U(\tilde{X}_i)$, where $V_i \sim N(0,I_d)$ is a Gaussian ``velocity" term, and $\eta>0$ is a step-size. 
  At each step, the unadjusted Langevin algorithm chain accumulates some error in its approximation of the Langevin diffusion.  To sample with some given accuracy $\epsilon$, the step size $\eta$ should be chosen small enough so that the total error accumulated  by the time the Langevin diffusion has reached a new roughly independent point is no more than $\epsilon$.

The Metropolis-adjusted Langevin algorithm (MALA) avoids the accumulation of error by introducing a Metropolis correction step. 
  The Metropolis correction step ensures that the MALA Markov chain has the correct stationary distribution.  For this reason, MALA does not need to approximate a particular outcome of the Langevin diffusion process in order to sample from the correct stationary distribution. Instead, $\eta$ only needs to be set small enough that each individual step of the MALA Markov chain has a high enough (in practice, $\Omega(1)$) acceptance probability. 
  In many situations, this lack of error accumulation is thought to allow MALA to take longer steps than ULA while still sampling from the correct stationary distribution (\cite{roberts1998optimal}).   

Another advantage of the Metropolis correction is that it allows the MALA Markov chain to converge exponentially quickly to the target distribution, meaning that MALA can sample with accuracy $\epsilon$ in a number of steps that depends logarithmically on $\epsilon^{-1}$.  ULA, on the other hand, requires a step size that is polynomial in $\epsilon^{-1}$ to approximate the Langevin diffusion with accuracy $\epsilon$.  
This logarithmic dependence on  $\epsilon^{-1}$ was shown in \cite{dwivedi2018log} to hold in the special case when the target distribution is strongly logconcave.

In the case of MALA, the proposal step is $\hat{X}_{i+1} = X_i + \eta V_i- \frac{1}{2}\eta^2 \nabla U(X_i)$, and the Metropolis correction step is  $\min(e^{\mathcal{H}(\hat{X}_{i+1},\hat{V}_{i+1})- \mathcal{H}(X_i,V_i)},1)$, where $\hat{V}_{i+1} = V_i +\frac{1}{2}\eta \nabla U(X_i)  - \frac{1}{2}\eta \nabla U(\hat{X}_{i+1})$.  The Hamiltonian functional $\mathcal{H}$ is defined as $\mathcal{H}(x,v) := U(x) + \mathcal{K}(v)$, where $U(x)$ is the ``potential energy" of a particle and $\mathcal{K}(v) = \frac{1}{2}\|v\|^2$ is its ``kinetic energy" (see for instance \cite{neal2011mcmc}).   The pair $(\hat{X}_{i+1}, \hat{V}_{i+1})$  approximates the position and velocity of a particle in classical mechanics with initial position $X_i$ and initial velocity $V_i$; this approximation is referred to as the ``leapfrog integrator" and is known to be a second-order method (that is, the error scales as $\eta^3$ in the limit as $\eta \downarrow 0$).  
  The acceptance probability for MALA therefore measures the extent to which our approximation of the particle's trajectory conserves the Hamiltonian. 

\paragraph{Our contributions.}
In this paper we obtain improved mixing time bounds for the Metropolis-adjusted Langevin algorithm.  In particular, to obtain faster bounds, we use the fact that the velocity term $V_i$ in the MALA algorithm points in a random direction.  Since the Hamiltonian changes much more quickly when the velocity term points in a worst-case direction than in a typical random direction, bounding the change in the Hamiltonian for this ``average-case" velocity in many cases allows us to use a relatively larger step size than would be possible using a worst-case analysis, while still having an O(1) acceptance probability.  This is in contrast to previous analyses of Langevin-based algorithms  (\cite{raginsky2017non,  zhang2017hitting, dwivedi2018log, cheng2017convergence}), whose bounds are obtained by assuming that $V_i$ travels in the worst-case direction at every step.
  
We bound the change in the Hamiltonian as a function of the third and fourth derivatives of $U$. Our bounds rely on the fact that in many applications the third derivative $\nabla^3 U(x)[V_i,V_i,V_i]$ and fourth derivative $\nabla^4 U(x)[V_i,V_i,V_i, V_i]$ are much larger if  $V_i$ points in the worst-case directions than if it points in a typical random direction.  We obtain bounds in terms of regularity constants $C_3$ and $C_4$, which, roughly speaking, bound these derivatives of $U$ as a function of $\|\mathsf{X}^\top V_i\|_\infty$.  The columns of the matrix $\mathsf{X}$ represent the ``bad" directions in which the potential function has larger higher-order derivatives.  For instance, in Bayesian logistic regression, these directions correspond to the independent variable data vectors.  Since $V_i \sim N(0,I_d)$, the velocity $V_i$ is unlikely to have a large component in any of these bad directions, meaning that $\|\mathsf{X}^\top V_i\|_\infty$ in many cases is much smaller than the Euclidean norm $\|V_i\|_2$.

The regularity condition for the third derivative is similar to the condition introduced in \cite{mangoubi2018dimensionally} to analyze the Hamiltonian Monte Carlo algorithm in the special case when the log-density $U$ is strongly convex.  However, in this paper, we prove bounds for the more general case when $U$ may be weakly convex or even non-convex.  To obtain these bounds in this more general case, we use the conductance method.  This allows us to bound the mixing time of MALA as a function of the Cheeger constant $\psi_\pi$  of the (possibly nonconvex) target log-density.  For many distributions, our bounds are faster than the current best bounds for the problem of sampling from these distributions.  For instance, when $\pi$ is weakly log-concave with identity covariance matrix, the log density has  $M$-Lipschitz gradient with $M=O(1)$, third-order smoothness \footnote{See Assumption \ref{assumption:derivatives} for a detailed definition of the smoothness constants $C_3$ and $C_4$.} $C_3 = O(\sqrt{d})$, and fourth-order smoothness $C_4 = O(d)$, we show that MALA can sample with TV accuracy $\epsilon$ in $d^{\frac{7}{6}} \log(\frac{\beta}{\epsilon})$ gradient evaluations given a $\beta$-warm start \footnote{We say $X_0$ is a $\beta$-warm start if it is sampled from a distribution $\mu_0$ where $\sup_{S\subseteq \mathbb{R}^d} \frac{\mu_0(S)}{\pi(S)} \leq \beta$.} (Section \ref{thm:main}), improving in this setting on the previous best bound of $d^{2.5} \log(\frac{1}{\epsilon})$ function evaluations which were obtained for the Random walk Metropolis (RWM) algorithm (\cite{lee2017eldan}).  As one concrete application, we show that MALA can sample in $d^{\frac{7}{6}} \log(\frac{\beta}{\epsilon})$ gradient evaluations for a class of Bayesian logistic regression problems with weakly convex priors, obtaining the fastest bounds for this class of problems (Section \ref{sec:applications}).  More generally, for these values of $M$, $C_3$, and $C_4$, we show that the number of gradient evaluations required to sample from possibly nonconvex targets is $d^{\frac{2}{3}}\psi_\pi^{-2} \log(\frac{\beta}{\epsilon})$.  For this setting our bounds for MALA are faster than the $\psi_{\pi}^{-10} d^{10} \log^5(\frac{1}{\epsilon})$ bounds of \cite{raginsky2017non} for the Stochastic gradient Langevin dynamics algorithm, as well as the best current bound of $d^{2}\psi_\pi^{-2}  \log(\frac{1}{\epsilon})$ for RWM in this setting, which we formally prove in Section \ref{sec:RWM}. 

We also prove related bounds when MALA is used as an optimization technique.  Our bounds for the optimization problem are given in terms of the restricted Cheeger constant, which was first introduced in \cite{zhang2017hitting}.  As one application, we obtain the fastest running time bounds for the zero-one loss minimization problem analyzed in both \cite{awasthi2015efficient} and \cite{zhang2017hitting} (Section \ref{sec:applications}).


\section{Previous results}

\paragraph{Previous results for sampling.}
In the setting where $U$ is (weakly) convex, \cite{lee2017eldan} show that one can sample with TV error $\epsilon$ in $O(d^{2.5} \log(\frac{\beta}{\epsilon}))$ function evaluations from a $\beta$-warm start if the target distribution $\pi$ is in isotropic position (that is, it has covariance matrix where the ratio of the largest to smallest eigenvalue is $O(1)$).  \cite{durmus2016sampling2} and \cite{dalalyan2017theoretical}  show that one can sample from a weakly log-concave distribution with $d^3 \epsilon^{-4}  \log(\frac{\beta}{\epsilon})$ gradient evaluations with the unadjusted Langevin algorithm (ULA) (see also \cite{cheng2017convergence} \footnote{ \cite{cheng2017convergence} show that ULA can sample in $\frac{d M^2 \hat{\beta}^4}{\epsilon^6}$ gradient evaluations, if given a ``Wasserstein warm start" $\mu_0$ such that $W_2(\mu_0, \pi) \leq \hat{\beta}$, and $U$ is $M$-smooth.  If the target density is in isotropic position, and given a $\beta$-warm start and exponential tails with $\mathsf{a} = \Omega(1)$, we have $\hat{\beta} = O(\sqrt{d} \log(\beta))$, meaning that the bound in \cite{cheng2017convergence} gives $O(d^3 \epsilon^{-6} \log^4(\beta))$ gradient evaluations for the usual warm start if $M=O(1)$.}).   \cite{dwivedi2018log} also analyze the MALA algorithm in the weakly log-concave setting, and obtain a bound of $O(d^3 \epsilon^{-1.5}) \log(\frac{\beta}{\epsilon})$,  if $M=O(1)$ and the fourth moments of $U$ are bounded by $\nu = O(d^2)$.

In the setting where $U$ is non-convex, \cite{raginsky2017non} show that the stochastic gradient Langevin dynamics algorithm can sample with Wasserstein error $\epsilon$ in 
$\tilde{O}([\lambda_{\pi}^{-1} \frac{M}{m} d ((b+d)M^2 + \sqrt{\sigma} M \sqrt{b+d}) \epsilon^{-4} \log(\frac{1}{\beta})]^5)$ stochastic gradient evaluations
 from a $\beta$-warm start, where $\lambda_{\pi}$ is the spectral gap of the Langevin diffusion on $U$, if $U$ is $(m,b)$-dissipative \footnote{$U$ is $(m,b)$-dissipative if $\nabla U(x)^\top x \geq m\|x\|_2^2 - b$} and the variance of the stochastic gradient is bounded by $\sigma^2 M^2 \|x\|_2^2$.  \cite{raginsky2017non} show that $\lambda_{\pi}^{-1}$ is bounded above by the Poincar\'e constant.  Since the Poincar\'e constant is bounded above by $\psi_{\pi}^{-2}$, this gives $\lambda_{\pi}^{-1} \leq \psi_\pi^{-2}$ (\cite{ledoux2000geometry}).  Therefore, in terms of the Cheeger constant, their bound gives $\tilde{O}([\psi_\pi^{-2} \frac{M}{m} d ((b+d)M^2 + \sqrt{\sigma} M \sqrt{b+d}) \epsilon^{-4} \log(\frac{1}{\beta})]^5)$.  See also \cite{bou2013nonasymptotic} for geometric ergodicity results for MALA, and \cite{eberle2014error} for an analysis of MALA on logdensities which are strongly convex outside a ball centered at the minimizer of the logdensity.

\begin {table}[t]
\begin{tabular}{ | l | c | c | c |}
  \hline			
   & \# of (stochastic) gradient  \\
   & or function calls   \\
  \hline
      Hit-and-run, \cite{lovasz2003hit, lovasz2006fast} & $d^{3} \log(\frac{\beta}{\epsilon})$  \\
  Ball walk or RWM, \cite{lee2017eldan}  & $d^{2.5} \log(\frac{\beta}{\epsilon})$   \\
  ULA, \cite{durmus2016sampling2}, \cite{dalalyan2017theoretical} & $d^3 \epsilon^{-4}  \log(\frac{\beta}{\epsilon})$ \\
  MALA, \cite{dwivedi2018log} &   $d^3 \epsilon^{-1.5} \log(\frac{\beta}{\epsilon})$ \\
  MALA, this paper &   $\max \left(C_3^{\frac{2}{3}}d^{\frac{5}{6}} , d^{\frac{7}{6}}, C_4^{\frac{1}{2}} d^{\frac{1}{2}}\right) \log(\frac{\beta}{\epsilon})$  \\

  \hline  
\end{tabular}
\caption{Number of gradient or function evaluations to sample from a weakly log-concave distribution with TV error $\epsilon$, with $\beta$-warm start, if target density has identity covariance matrix. For simplicity, we assume that $\pi$ has exponential tails with decay rate $\Omega(\frac{1}{\sqrt{d}})$, and that $M, \nu= O(1)$.}\label{table:WeakConvexity}
\end{table}

\begin {table}[t]
\begin{tabular}{ | l | c | c |  c |}
  \hline			
   & \# of (stochastic) gradient & \# Markov chain   & mode of \\
   & or function calls & steps  & convergence  \\
  \hline
    ULA \cite{raginsky2017non}  & $\psi_{\pi}^{-10} m^{-5} d^{10} \log^5(\frac{\beta}{\epsilon})$  & same & Wasserstein  \\
    SGLD \cite{raginsky2017non}  & $\psi_{\pi}^{-10}  m^{-5} d^{10} \log^5(\frac{1}{\beta}) \times  (1 + \sqrt{\frac{\sigma}{d}})$  & same & Wasserstein \\
    RWM   [this paper]  & $d^2 \psi_{\pi}^{-2} \log(\frac{\beta}{\epsilon})$ & same & TV \\
    MALA [this paper]  & $\min \left(C_3^{\frac{2}{3}}d^{\frac{1}{3}} , d^{\frac{2}{3}}, C_4^{\frac{1}{2}}\right) \psi_\pi^{-2} \log(\frac{\beta}{\epsilon})$  & same & TV\\
      \hline			
      RHMC Markov chain & Not an algorithm in this setting &  $d^{\frac{1}{2}}\tilde{\psi}_\pi^{-2} R \log(\frac{\beta}{\epsilon})$ & TV  \\
      \cite{lee2017convergence} & &   &   \\
  \hline  
\end{tabular}
\caption{Number of gradient (or stochastic gradient) evaluations to sample with TV error $\epsilon$, from a possibly nonconvex target distribution with Cheeger isoperimetric constant $\psi_\pi$, given a $\beta$-warm start.  R is a regularity parameter for $U$ with respect to the Riemmannian metric used by RHMC, and $\tilde{\psi}_\pi$ is an isoperimetric constant for the target $\pi$ with respect to this Riemmannian metric; note that $\tilde{\psi}_\pi$ is equal to $\psi_\pi$ when RHMC uses the Euclidean metric.  For simplicity, we assume in this table that $M=O(1)$ and that $\pi$ has exponential tails with decay rate $\Omega(\frac{1}{\sqrt{d}})$ (that is, $\mathsf{a}=\Omega(1)$ in Assumption \ref{assumption:tails}.  For ULA and SGLD, we assume that $\pi$ is $(m,b)$-dissipative with $b = O(d)$.).}\label{table:WeakConvexity}
\end{table}

\paragraph{Previous results for nonconvex optimization.}

One can also consider the problem of optimizing a function $F: \mathbb{R}^d \rightarrow \mathbb{R}$ on some subset $\mathsf{S} \subseteq \mathbb{R}^d$. \cite{raginsky2017non} show that they can obtain an $\tilde{O}(\frac{(\epsilon+\sqrt{\sigma}) d^2}{\psi_{\pi}^2} +d)$-approximate minimizer in $\tilde{O}(\frac{d}{\psi_{\pi}^2 \frac{M}{m}\epsilon^4})$ stochastic gradient evaluations.

 \cite{zhang2017hitting} show that, under certain assumptions on the constraint set $\mathsf{S}$, given a $\beta$-warm start, the stochastic gradient Langevin dynamics algorithm can be used to obtain an approximate minimizer $\hat{x}^\star$ such that $F(\hat{x}^\star)-\min_{x\in \mathsf{S}}F(x) \leq \epsilon$ with probability at least $1-\delta$ in $d^4 \hat{\psi}^{-4} (G^4 +M^2) \log(\frac{\beta}{\delta})$ stochastic gradient evaluations. The quantity $\hat{\psi} \equiv \hat{\psi}_{e^{-F}}(\mathsf{S} \backslash \mathcal{U})$, is the ``restricted" version of the Cheeger constant for the log-density $F$, restricted to the set $\mathsf{S} \backslash \mathcal{U}$, where  $\mathcal{U}$ is a set consisting of only $\epsilon$-approximate minimizers of $F$, and $G^2$ is a bound on the variance of the stochastic gradient.

\section{Algorithms}

\subsection{Sampling algorithm}
We now state the usual version of the MALA algorithm which is used for sampling:
\begin{algorithm}[h]
\caption{MALA for sampling \label{alg:sampling}}
\textbf{input:} First-order oracle for gradient $\nabla U$, step size $\eta>0$\\
 \textbf{input:}   Initial point $X_0 \in \mathbb{R}^d$.\\
 \textbf{output:} Markov chain $X_0, X_1, \ldots, X_{i_{\max}}$ with stationary distribution $\pi \propto e^{-U}$.
\begin{algorithmic}[1]

\For{$i=0$ to $i_{\mathrm{max}}-1$}
Sample $V_i \sim N(0,I_d)$.\\
Set $\hat{X}_{i+1} = X_i + \eta V_i - \frac{1}{2}\eta^2 \nabla U(X_i)$\\
Set $\hat{V}_{i+1} = V_i - \eta \nabla U(X_i)  - \frac{1}{2}\eta^2 \frac{\nabla U(\hat{X}_{i+1})  - \nabla U(X_i)}{\eta}$.\\
 Set \be
X_{i+1} &= \begin{cases}\hat{X}_{i+1} \qquad \textrm{ with probability } \min(1, \, e^{\mathcal{H}(\hat{X}_i,\hat{V}_i)-\mathcal{H}(X_i, V_i)})
&\\ X_i \qquad \textrm{ otherwise}  \end{cases}
\ee
\EndFor
\end{algorithmic}
\end{algorithm}

\noindent
Every time a proposal $\hat{X}_{i+1}$  is made, the MALA algorithm accepts the proposal with probability $\min(1, \, e^{\mathcal{H}(\hat{X}_{i+1},\hat{V}_{i+1})- \mathcal{H}(X_i,V_i)})$.  One way to view this acceptance rule is that it is simply the Metropolis-Hastings rule for this proposal, which causes the transition kernel $K$ of the Markov chain to satisfy the detailed balance equations $K(x,y) \pi(x) = K(y,x) \pi(y)$, ensuring that MALA has stationary distribution $\pi$.  

One can also interpret the Metropolis acceptance rule in a different way, inspired by classical mechanics, which is the approach we use to obtain our bounds in this paper.  In this view  $\mathcal{H}(x,v) := U(x) + \mathcal{K}(v)$ gives the energy of a particle with position $x$ and velocity $v$, where $U(x)$ is the ``potential energy" of the particle and $\mathcal{K}(v) = \frac{1}{2}\|v\|^2$ is its ``kinetic energy". The values of $\hat{X}_{i+1}$ $\hat{V}_{i+1}$ can be viewed as a second-order numerical approximation to the position and velocity of a particle in classical mechanics, with initial position and velocity $X_i,V_i$.  The continuous dynamics, determined by Hamilton's equations, conserve the Hamiltonian.  If $(\hat{X}_{i+1}, \hat{V}_{i+1})$ approximate the outcome of the continuous dynamics with low error, the acceptance probability will be $\Omega(1)$.  The goal is to choose $\eta$ as large as possible while still having an $\Omega(1)$ acceptance probability.

\subsection{Constrained optimization algorithm}
One can also use MALA for constrained optimization.  For instance, we apply MALA to constrained optimization in Algorithm \ref{alg:optimization}.
\begin{algorithm}[h]
\caption{MALA for constrained optimization \label{alg:optimization}}
\flushleft
\textbf{input:} zeroth-order oracle for $U:\mathbb{R}^d \rightarrow \mathbb{R}$, first-order oracle for gradient $\nabla U$, membership oracle for a constraint set $\mathsf{S} \subseteq \mathbb{R}^d$, step size $\eta>0$\\
 \textbf{input:}   Initial point $X_0 \in \mathbb{R}^d$.\\
 \textbf{output:} An approximate global minimizer $\hat{x}^\star \in \mathsf{S}$
 
\begin{algorithmic}[1]
\For{$i=0$ to $i_{\mathrm{max}}-1$}
Sample $V_i \sim N(0,I_d)$.\\
Set $\hat{X}_{i+1} = X_i + \eta V_i - \frac{1}{2}\eta^2 \nabla U(X_i)$\\
Set $\hat{V}_{i+1} = V_i - \eta \nabla U(x)  - \frac{1}{2}\eta^2 \frac{\nabla U(z)  - \nabla U(x)}{\eta}$.
\be
\textrm{Set }& \qquad \qquad Z_{i+1} = \begin{cases}\hat{X}_{i+1} \qquad \textrm{ with probability } \min(1, \, e^{\mathcal{H}(\hat{X}_i,\hat{V}_i)-\mathcal{H}(X_i, V_i)})
&\\ X_i \qquad \textrm{ otherwise}  \end{cases} \qquad \qquad \quad\\
\textrm{Set }&\qquad \qquad X_{i+1} = \begin{cases}Z_{i+1} \qquad \textrm{ if } Z_{i+1} \in \mathsf{S}
\\ X_i \qquad \textrm{ otherwise}  \end{cases}
\ee
\EndFor\\
Set $\hat{x}^\star = X_{i^\star}$, where $i^\star = \mathrm{argmin}_{i \in \{0,\ldots, i_{\mathrm{max}}\}} U(X_i)$
\end{algorithmic}
\end{algorithm}

\section{Assumptions and notation}
\subsection{Smoothness and tail bound assumptions}
In our main result we show that, under certain regularity conditions, MALA can sample from $O(d^{\frac{2}{3}}\psi_\pi^{-2} \log(\frac{\beta}{\epsilon}))$ gradient evaluations.  In this section we explain why these regularity conditions are needed to obtain bounds for MALA with dimension dependence smaller than $d^1$.

   We start by noting that if one attempts to bound the number of gradient evaluations required by MALA using a conventional Euclidean operator norm bound on the higher derivatives of $U$,  then the bounds that one obtains in terms of the Cheeger constant are no faster than $d \psi_\pi^{-2}$ gradient evaluations.
 Recall that $\hat{X}_{i+1}$ $\hat{V}_{i+1}$ can be viewed as a second-order numerical approximation to the $\hat{x}$ position and velocity $\hat{v}$ of a particle in classical mechanics after time $\eta$, which has initial position and velocity $X_i,V_i$.  Bounding the numerical error $\hat{X}_{i+1} - \hat{x}$ and $\hat{V}_{i+1} - \hat{v}$ gives us a bound on the Hamiltonian.  In particular, for the kinetic energy error we have:
 \be
 |\mathcal{K}(\hat{v}) - \mathcal{K}(\hat{V}_{i+1})| & \approx  |(\hat{V}_{i+1}- \hat{v})^\top \nabla \mathcal{K}(\hat{v})| = |(\hat{V}_{i+1}- \hat{v})^\top \hat{v}|\\
&\approx  |\int_0^\eta \int_0^r  V_i^\top [\nabla^2U(X_i)- \nabla^2U(X_i+V_i\tau)] V_i\mathrm{d}\tau \mathrm{d}r|\\
&\approx \left|\eta^3 \nabla^3U(X_i)[V_i,V_i, V_i] + \eta^4 \nabla^4U(X_i)[V_i,V_i,V_i,V_i] \right|.
\ee
If we assume the usual ``operator norm" Euclidean bound on $\nabla^3U$ and $\nabla^4U$, we have\\ $\eta^3 \nabla^3U(X_i)[V_i,V_i, V_i] \leq L_3 \eta^3 \|V_i\|_2^3$ and $\eta^4 \nabla^4U(X_i)[V_i,V_i,V_i,V_i] \leq \eta^4 L_4 \|V_i\|_2^4$ for some $L_3,L_4>0$.  Since $V_i \sim N(0,I_d)$, we have  $\|V_i\|_2 = \tilde{O}(\sqrt{d})$ with high probability.  Hence, to obtain an $O(1)$ bound on the kinetic energy error, we require $\eta = d^{-\frac{1}{2}}$ if $L_3, L_4 = \Theta(1)$.  Since the distance traveled by the MALA Markov chain after $i$ steps is roughly proportional to $\eta \sqrt{d} \sqrt{i}$, the number of steps to explore a distribution with most of the probability measure in a ball of diameter $\sqrt{d}$ is roughly  $i=d$ for this choice of $\eta$ if $\psi_\pi^{-1} = 1$ (for instance, this is the case when $\pi$ is a standard Gaussian, and $\psi_\pi^{-1}=1$ by the Gaussian isoperimetric inequality).

To obtain an $O(1)$ energy error for a larger step size $\eta$, we need to control $\nabla^3U(X_i)[V_i,V_i, V_i] $ and  $\nabla^4U(X_i)[V_i,V_i, V_i, V_i]$  with respect to a norm which does not grow as quickly with the dimension as the Euclidean norm for a random $N(0,I_d)$ velocity vector $V_i$.  One way to do so would be to replace these bounds with an infinity-norm condition 
$
\nabla^3U(X_i)[V_i,V_i,V_i] \leq C_3 \|V_i\|_\infty^3
$ 
and 
$
\nabla^4U(X_i)[V_i,V_i,V_i,V_i] \leq C_4 \|V_i\|_\infty^4.
$ 
  For this norm, $\|V_i\|_\infty =O(\log(d))$ with high probability since $V_i \sim N(0,I_d)$, implying that $\eta^3 \nabla^3U(X_i)[V_i,V_i, V_i] \leq C_3 \eta^3 \log^3(d)$ rather than $\eta^3 \nabla^3U(X_i)[V_i,V_i, V_i] \leq L_3 \eta^3 d^{\frac{3}{2}}$, and $\eta^4 \nabla^4U(X_i)[V_i,V_i, V_i, V_i] \leq C_4 \eta^4 \log^4(d)$ rather than\\ $\eta^4 \nabla^4U(X_i)[V_i,V_i, V_i, V_i] \leq L_4 \eta^4 d^{2}$.  Since for many distributions of interest this condition does not hold for small values of $C_3$ and $C_4$, we use a more generalize condition, to obtain smaller $C_3$ and $C_4$ constant for a wider class of distributions.  Specifically, we replace the norm  $\|V_i\|_\infty$  with a more general norm  $\| \mathsf{X}^\top V_i\|_\infty$ for some matrix $\mathsf{X}$.  Roughly speaking, this regularity condition allows the third and fourth derivatives to be large in $r>0$ ``bad" directions $\mathsf{X}_1, \ldots, \mathsf{X}_r$, as long as they are small in a typical random direction.  More specifically, we assume that
  \begin{assumption}[$C_3, C_4>0, \mathsf{X} = {[}\mathsf{X}_1,\ldots, \mathsf{X}_r{]}$ where  $\|\mathsf{X}_i\|_2 = 1$ for all $i \in {[r]}$]\label{assumption:derivatives}
For all $x, u,v,w \in \mathbb{R}^d$, we have
\be
| \nabla^3U(x)[u,v, w] | &\leq C_3 \|\mathsf{X}^\top u\|_{\infty} \|\mathsf{X}^\top v\|_{\infty}  \|w\|_{2},\\
| \nabla^4U(x)[u,u, u, u] | &\leq C_4 \|\mathsf{X}^\top u\|_{\infty}^4.
\ee
\end{assumption}
We expect this assumption to hold with relatively small values of $C_3$ and $C_4$ when the target function $U$ is of the form $U(x) = \sum_{i=1}^r f_i(u_i^\top x)$ for functions $f_i:\mathbb{R}\rightarrow \mathbb{R}$ with uniformly bounded third and fourth derivatives.  In particular, this class includes the target functions used in logistic regression as well as smoothed versions of the nonconvex target functions used when learning linear classifiers with zero-one loss.  Finally, we note that our assumption on $\nabla^3U$ includes both infinity norms and a Euclidean norm, since our rough approximation of the error in this section ignores higher-order terms which are best bounded with a slightly different assumption that incorporates both norms.
\begin{remark}
Assumption \ref{assumption:derivatives} has two infinity-norms on the right hand side, and one Euclidean norm.  One could instead make a strictly stronger assumption which instead has three infinity norms.  It is an interesting open question whether this stronger assumption would lead to an even stronger bound on the number of gradient evaluations in special cases.
\end{remark}

We also make the assumption that the target distribution $\pi$ has exponential tails (here $x^\star$ is a global minimizer of $U$ on $\mathbb{R}^d$):
\begin{assumption}[exponential tail bounds ($\mathsf{a}>0$)] \label{assumption:tails}
Suppose that $X\sim \pi$.  Then $\mathbb{P}(\|X -x^\star\|_2> s) \leq e^{-\frac{\mathsf{a}}{\sqrt{d}} s}$.
\end{assumption}
We also assume that $U$ has Lipschitz gradient
\begin{assumption}[Lipschitz gradient ($M \geq 0$)] \label{assumption:gradient}
For all $x \in \mathbb{R}^d$ we have $\|\nabla U(x)\|_2 \leq M$.
\end{assumption}

For the problem of constrained optimization on a subset $\mathsf{S} \subseteq \mathbb{R}^d$, we make the following regularity assumption on $\mathsf{S}$:
\begin{assumption}[Constraint set exit probability] \label{assumption:constraint}
For any $z\in \mathsf{S}$, let $\gamma_z := z+ \eta v - \frac{1}{2} \eta^2 \nabla U(x)$ where $v \sim N(0,I_d)$.  We assume that $\mathbb{P}(\gamma_z \in \mathsf{S}) \geq \frac{1}{10} \qquad \forall z \in \mathsf{S}$.
\end{assumption}
\begin{remark}
We note that Assumption \ref{assumption:tails} always holds for some value of $\mathsf{a}>0$ if the target distribution is logconcave.  Specifically, a logconcave probability distribution must integrate to 1, and have convex sublevel sets, implying that these level sets must be compact.  Let $h_0$ be the height of the maximizer of the target density.  For the log-density to integrate to 1, one must have a compact sublevel set with height $h$ strictly less than $h_0$, bounded by a ball of some radius $r$.  By convexity of the log-density, the decay rate is at least $\mathsf{a} \geq (h_0 - h)/r$.
\end{remark}

\subsection{Cheeger constants}\label{sec:Cheeger}
\vspace{-2mm}
For any set $A\subset \mathbb{R}^d$, define $A_\epsilon := \{x\in \mathbb{R}^d : \inf_{y \in A} \|x - y\|_2 \leq \epsilon\}$.
We define the Cheeger constant $\psi_{\pi}$ of a distribution $\pi$ with support $\mathsf{S}\subseteq \mathbb{R}^d$ as follows: $\psi_{\pi} := \liminf_{\epsilon \downarrow 0} \inf_{S  \subseteq \mathsf{S} \, : \, 0 < \pi(S) < \frac{1}{2}}   \frac{ \pi(S_\epsilon) - \pi(S)}{ \epsilon \pi(S) }$.
For any Markov chain with transition kernel $K$ and stationary distribution $\pi$, we define the conductance $\Psi_{K}$ of the Markov chain to be: $\Psi_{K} := \inf_{S  \subseteq \mathsf{S} \, : \, 0 < \pi(S) < \frac{1}{2}}   \frac{ K(S,\mathsf{S} \backslash S) }{\pi(S) }$.

Next, for any $V\subseteq \mathbb{R}^d$ we define the ``restricted Cheeger constant," originally introduced in \cite{zhang2017hitting}, as $\hat{\psi}_\pi(V) := \liminf_{\epsilon \downarrow 0} \inf_{S \subseteq V  \, : \, \pi(S)>0}  \frac{ \pi(S_\epsilon) - \pi(S)}{\epsilon \pi(S) }$,
and the restricted conductance $\hat{\Psi}_{K} := \inf_{S  \subseteq V \, : \,  \pi(S)>0}   \frac{ K(S,\mathsf{S} \backslash S) }{\pi(S) }$.

\subsection{Other Notation}
\vspace{-3mm}
We say $X_0$ is a $\beta$-warm start if it is sampled from a distribution $\mu_0$ where $\sup_{S\subseteq \mathbb{R}^d} \frac{\mu_0(S)}{\pi(S)} \leq \beta$.  For any probability distribution $\mu: \mathbb{R}^d \rightarrow \mathbb{R}$ denote $\Sigma_{\mu}$ the covariance matrix of the distribution $\mu$.  We denote the $d\times d$ identity matrix by $I_d$.  For any subset $\mathcal{U} \subseteq \mathbb{R}^d$ and $\Delta>0$, we define the $\Delta$-thickening of $\mathcal{U}$ by $\mathcal{U}_\Delta := \{x\in \mathbb{R}^d : \inf_{y\in\mathcal{U}} \|y-x\|_2 \leq \Delta\}$. We denote the total variation norm of a measure $\mu$ by $\|\mu\|_{\mathrm{TV}} := \sup_{S\subseteq \mathbb{R}^d} \mu(S)$.  For any random variable $Z$, let $\mathcal{L}(Z)$ denote the distribution of this random variable.

\section{Main results}
\vspace{-2mm}
\subsection{Main Theorems for sampling and optimization}
First, we state our main theorem for the sampling problem:
\begin{theorem} [sampling]\label{thm:main}
Suppose that $U$ satisfies Assumptions \ref{assumption:derivatives} and \ref{assumption:tails}, and has $M$-Lipschitz gradient on $\mathbb{R}^d$. %
 Then given a $\beta$-warm start, for any step-size parameter\\ $\eta \leq \tilde{O}\left(\min \left(C_3^{-\frac{1}{3}}d^{-\frac{1}{6}} , d^{-\frac{1}{3}}, C_4^{-\frac{1}{4}}\right) \min(1, M^{-\frac{1}{2}}) [\log \log(\frac{1}{\mathsf{a}})]^{-1}\right)$ there exists $\mathcal{I} = O( ((\eta^{-1} + \eta M)\psi_\pi)^{-2} \log(\frac{\beta}{\epsilon}))$ for which $X_i$ of Algorithm \ref{alg:sampling} satisfies $\| \mathcal{L}(X_i) - \pi\|_{\mathrm{TV}} \leq \epsilon$ for all $i\geq \mathcal{I}$.
\end{theorem}
\noindent Theorem \ref{thm:main} states that, from a $\beta$-warm start, the MALA Markov chain generates a sample from $\pi$ with TV error $\epsilon$ in $O( ((\eta^{-1} + \eta M)\psi_\pi)^{-2} \log(\frac{\beta}{\epsilon}))$ gradient evaluations if $U=-\log(\pi)$ satisfies Assumptions \ref{assumption:derivatives} and \ref{assumption:tails} and has $M$-Lipschitz gradient (Assumption \ref{assumption:gradient}).  Recall from Section \ref{sec:Cheeger} that $\psi_\pi$ is the Cheeger constant of $\pi$.  In particular, when $\pi$ is weakly log-concave with identity covariance matrix, we have that $\psi_{\pi} = \Omega(d^{-\nicefrac{1}{4}})$ by Theorem 7 in \cite{lee2017eldan}.  If we also have that the log-density has  $M$-Lipschitz gradient with $M=O(1)$, third-order smoothness $C_3 = O(\sqrt{d})$, and fourth-order smoothness $C_4 = O(d)$, then MALA can sample with TV accuracy $\epsilon$ in $d^{\frac{7}{6}} \log(\frac{\beta}{\epsilon})$ gradient evaluations given a $\beta$-warm start.

Next, we state our main theorem for the problem of optimizing a function on a subset $\mathsf{S}\subset \mathbb{R}^d$: 

\begin{theorem} [optimization]\label{thm:main_optimization}
Suppose that $U:\mathbb{R}^d \rightarrow \mathbb{R}$ satisfies Assumptions \ref{assumption:derivatives} and \ref{assumption:tails}, and has $M$-Lipschitz gradient on $\mathbb{R}^d$, and that $\mathsf{S} \subseteq \mathbb{R}^d$ satisfies Assumption \ref{assumption:constraint}. %
Choose a step-size $\eta \leq \tilde{O}\left(\min \left(C_3^{-\frac{1}{3}}d^{-\frac{1}{6}} , d^{-\frac{1}{3}}, C_4^{-\frac{1}{4}}\right) \min(1, M^{-\frac{1}{2}})  [\log \log(\frac{1}{\mathsf{a}})]^{-1}\right)$ in Algorithm \ref{alg:optimization}.   
   Let $\pi(x) \propto e^{-U(x)} \mathbbm{1}_{\mathsf{S}}$ and %
     let $\mathcal{U}\subseteq \mathsf{S}$. Then given an initial point which is $\beta$-warm with respect to $\pi$, for any $\delta>0$ we have %
$
\inf\{i :X_i \in \mathcal{U}_\Delta\} \leq \mathcal{I}
$
 with probability at least $1- \delta$, where $\mathcal{I} = \frac{4 \log(\frac{\beta}{\delta})}{\Delta^2 \hat{\psi}_{\pi}^2(\mathsf{S} \backslash \mathcal{U})}$ and $\Delta = \frac{1}{100}(\frac{1}{2}\eta^{-1} + \frac{1}{4} \eta M)^{-1}$.
\end{theorem}
\noindent Theorem \ref{thm:main_optimization} states that, if $U$ satisfies the higher-order smoothness Assumptions \ref{assumption:derivatives} and \ref{assumption:tails}, has $M$-Lipschitz gradient (Assumption \ref{assumption:gradient}), and the constraint set $\mathsf{S}$ satisfies Assumption \ref{assumption:constraint}, then, roughly speaking, one can find an approximate minimizer for $U$ on a subset $\mathsf{S}$.  More specifically, if $U(x)$ is $R$-Lipschitz on $\mathsf{S}$, and we take $\mathcal{U}$ to be the sublevel set $\mathcal{U} = \{x \in \mathsf{S} : U(x) \leq \epsilon \inf_{y\in \mathsf{S}} U(y)\}$ consisting of $\epsilon$-minimizers of $U$ on $\mathsf{S}$ and one chooses $\eta$ small enough that $\Delta \leq \frac{\epsilon}{R}$, then Theorem \ref{thm:main_optimization} says that the number of gradient evaluations to obtain a $2\epsilon$-minimizer of $U$ is bounded by $O(\frac{4 \log(\frac{\beta}{\delta})}{\Delta^2 \hat{\psi}_{\pi}^2(\mathsf{S} \backslash \mathcal{U})})$.  In Section \ref{sec:applications} we apply Theorem \ref{thm:main_optimization} to obtain an improved bound on the number of gradient and function evaluations for a class of non-convex optimization problems for Linear classifiers with binary loss (Theorem \ref{thm:ZeroOneLoss}).

\subsection{Applications} \label{sec:applications}
 \paragraph{Applications to Bayesian regression.}
In Bayesian regression, one would like to sample from the target log-density 
$U(\theta) = F_0(\theta)- \textstyle{\sum_{i=1}^r} \mathcal{Y}_i \varphi(\theta^\top \mathcal{X}_i) + (1-\mathcal{Y}_i)\varphi(-\theta^\top \mathcal{X}_i)$,
 where the data vectors $\mathcal{X}_1,\ldots \mathcal{X}_r \in \mathbb{R}^d$ are thought of as independent variables, the binary data $\mathcal{Y}_1,\ldots, \mathcal{Y}_r \in \{0,1\}$ are dependent variables, $\varphi:\mathbb{R} \rightarrow \mathbb{R}$ is the loss function, and $F_0$ is the Bayesian log-prior.  We will assume that $\varphi$ has its first four derivatives uniformly bounded by 1. Two smooth loss functions of interest in applications are the (convex) logistic loss function $\varphi(s) = -\log (e^{-s}+1)^{-1}$ used in logistic regression, and the non-convex sigmoid loss function $\varphi(s) =(e^{-s}+1)^{-1}$ which is more robust to outliers.
 We define the incoherence of the data as  ${\textstyle \mathsf{inc}(\mathcal{X}_1,\ldots \mathcal{X}_r):= \max_{i\in[r]}  \sum_{j=1}^r |\mathcal{X}_i^\top \mathcal{X}_j|.}$
   We bound the value of the constant $C_3$ in terms of the incoherence: 
\begin{theorem}[\bf Regularity bounds for empirical functions, Th. 2 of \cite{mangoubi2018dimensionally}] \label{thm:logit}
Let $U(x) = F_0(x) + \textstyle{\sum_{i=1}^r} \mathcal{Y}_i \hat{\varphi}(\theta^\top \mathcal{X}_i) + (1-\mathcal{Y}_i)\hat{\varphi}(-\theta^\top \mathcal{X}_i)$, where $\varphi: \mathbb{R} \rightarrow \mathbb{R}$ is a function that satisfies $|\varphi'''(x)| \leq 1$, and $F_0$ is a quadratic function.  Let $\mathsf{inc}(\mathcal{X}_1,\ldots,\mathcal{X}_r)\leq \Phi$ for some $ \Phi >0$. 
 Then Ass. \ref{assumption:derivatives} is satisfied with $C_3 = \sqrt{r} \sqrt{ \Phi}$ and ``bad" directions $\mathsf{X}_i = \frac{\mathcal{X}_i}{\|\mathcal{X}_i\|_2} $, and with $C_4 \leq r$.
\end{theorem}

The proof of Theorem \ref{thm:logit} for the bound on $C_3$ is given in the arXiv version of \cite{mangoubi2018dimensionally}; see Section \ref{sec:C4} for the bound on $C_4$.

As an example, consider the case when all $r=\Theta(d^2 \log(\frac{d}{\delta}))$ unit vectors are isotropically distributed, and we have an improper prior, that is,  $F_0 = 0$.  Since $F_0 = 0$, the target distribution is not strongly log-concave; it is only weakly log-concave.  Suppose that $\|\theta^\star\|_2 = O(1)$. Since the vectors are isotropically distributed, with probability $1-\delta$ the covariance matrix $\Sigma_{\pi}$ of the distribution $\pi$ satisfies $c_1 \frac{d}{r}I_d \preccurlyeq \Sigma_{\pi} \preccurlyeq c_2 \frac{d}{r}I_d$ for some universal constants $c_1, c_2$ (see for instance the Matrix Chernoff inequality in \cite{tropp2012user} for the upper bound on the eigenvalues, and Lemma 9.4 of \cite{LeeMangoubiVishnoi18} for the lower bound on the eigenvalues). We can precondition $\pi$ by replacing $U(x)$ with the log-density $U(x) \leftarrow U(\frac{\sqrt{r}}{\sqrt{d}}x)$ and sampling from the distribution $\pi(x) \leftarrow \frac{e^{-U(x)}}{\int_{\mathbb{R}^d} e^{-U(x)} \mathrm{d} x}$; the covariance matrix of this preconditioned $\pi$ now satisfies $c_1 I_d \preccurlyeq \Sigma_{\pi} \preccurlyeq c_2 I_d$, implying that $\psi_{\pi} = \Omega(d^{-\nicefrac{1}{4}})$ by Theorem 7 in \cite{lee2017eldan}.    For this preconditioned $U$, we have $C_3 = O(1)$ and $C_4 = O(1)$, implying that by Theorem \ref{thm:main} we require at most $O(d^{\frac{2}{3}} \psi_{\pi}^{-2} \log(\frac{\beta}{\epsilon})) = O(d^{\frac{7}{6}} \log(\frac{\beta}{\epsilon}))$ gradient evaluations to sample with TV error $\epsilon$.  In this case we therefore have an improvement on the previous best bound for the non-strongly logconcave setting, proved for the ball walk Markov chain, which requires $O(d^{2} \psi_{\pi}^{-2} \log(\frac{\beta}{\epsilon})) = O(d^{2.5} \log(\frac{\beta}{\epsilon}))$ gradient evaluations (\cite{lee2017eldan}) (note, however, that this bound for the ball walk holds more generally for any log-concave distribution with identity covariance matrix).
\begin{remark}
We note that in this example one must compute the gradients of $r = \Theta(d^2 \log(\frac{d}{\delta}))$  component functions $\varphi((\frac{\sqrt{r}}{\sqrt{d}}x)^\top \mathcal{X}_i)$ in order to compute the gradient of $\nabla U(x)$. Therefore, it may be possible to improve on our dependence on $d$ by using a stochastic gradient-based method.  However, if one uses a stochastic gradient method, which lacks a Metropolis step, the dependence of the gradient evaluation bounds on $\epsilon^{-1}$ would no longer be logarithmic and would instead be polynomial.
\end{remark}
\begin{remark}
We note that whenever one can bound the Cheeger constant of a distribution $\pi(x) \propto e^{-U(x)}$, one also obtains a bound on the Cheeger constant of a perturbation of this distribution $\hat{\pi}(x) \propto e^{-U(x) +\phi(x)}$, if the perturbation $\psi:\mathbb{R}^d \rightarrow \mathbb{R}$ is uniformly bounded by some $b\geq 0$ \cite{applegate1991sampling}.  Namely the Cheeger constant of the perturbed distribution satisfies $\psi_{\hat{\pi}} \geq e^{2b} \psi_\pi$.  Therefore, our bounds of this section also apply, up to an $O(1)$ factor, to any nonconvex logdensities which are perturbations of the weakly convex logdensities analyzed in this section, provided the perturbations have magnitude $b = O(1)$.
\end{remark}

\paragraph{Linear classifiers with binary loss.}
In \cite{awasthi2015efficient} and \cite{zhang2017hitting} the authors study the problem of learning linear classifiers with zero-one loss functions. The goal is to estimate an unknown parameter $\theta^\star$, from data vectors $\mathcal{X}_1,\ldots \mathcal{X}_r \in \mathbb{R}^d$ that are thought of as independent variables, and binary response data $\mathcal{Y}_1,\ldots, \mathcal{Y}_r \in \{-1,1\}$. Here $(\mathcal{X}_i, \mathcal{Y}_i)$ are drawn i.i.d. from some probability distribution $\mathcal{P}$.
More specifically, the response variable in their model satisfies
\be
\mathcal{Y}_i = \begin{cases} \mathrm{sign}(\mathcal{X}_i^\top \theta^\star) \qquad \mathrm{w.p.} \qquad \frac{1+ \mathfrak{q}(\mathcal{X}_i)}{2}\\   -\mathrm{sign}(\mathcal{X}_i^\top \theta^\star) \qquad \qquad \qquad \textrm{otherwise}   \end{cases}
\ee
where $\mathfrak{q}: \mathbb{R}^d \rightarrow [0,1]$.  Here $\mathfrak{q}$ is assumed to satisfy $\mathfrak{q}(x) \geq \mathfrak{q}_0 |  x^\top \theta^\star |$ for some $\mathfrak{q}_0>0$. 
 \cite{awasthi2015efficient} and \cite{zhang2017hitting} consider the case where the $r$ data vectors are i.i.d. uniformly distributed on the unit sphere, with $r \geq \frac{d^4}{\mathfrak{q}_0^2 \epsilon^4}$.
 
The goal is to find an estimate for $\theta^\star$ which (approximately) minimizes the following population expected loss function: $F(x) := \mathbb{E}_{(a, b) \sim \mathcal{P}} \ell(x ; (a, b))$.
To find this estimate, \cite{zhang2017hitting} employ a stochastic gradient Langevin dynamics method, to obtain an approximate minimizer $\hat{\theta}$ such that $F(\hat{\theta}) - F(\theta^\star) < \epsilon$ with probability at least $1-\delta$ in $\tilde{O}(\frac{d^{13.5}}{\epsilon^{16}} \log(\frac{\beta}{\delta}))$ inner product evaluations, and  $\tilde{O}(\frac{d^{14.5}}{\epsilon^{16}} \log(\frac{\beta}{\delta}))$ arithmetic operations given a $\beta$-warm start, if $\mathfrak{q}_0 = O(1)$ \footnote{Each inner product takes $d$ arithmetic operations to perform.}.  We instead use Algorithm \ref{alg:optimization}, and show that one can use this algorithm to obtain an approximate minimizer in $\tilde{O}\left(d^{\frac{25}{6} + 4} \epsilon^{-\frac{22}{3} -4} \log(\frac{\beta}{\delta}) \log(\frac{1}{\delta}) \right)$ inner-product evaluations and $\tilde{O}\left(d^{\frac{25}{6} + 5} \epsilon^{-\frac{22}{3} -4} \log(\frac{\beta}{\delta}) \log(\frac{1}{\delta}) \right) \leq \tilde{O}\left(d^{9.2} \epsilon^{-11.4} \log(\frac{\beta}{\delta}) \log(\frac{1}{\delta}) \right)$ arithmetic operations.  This improves on the dependence of the previous best bound on $d$ and $\epsilon$, at the expense of a $\log(\frac{1}{\delta})$ factor.

To obtain their result, \cite{zhang2017hitting} attempt to find an approximate minimizer for the zero-one empirical risk function $f(x) := \sum_{i=1}^r \ell(x ; (\mathcal{X}_i, \mathcal{Y}_i))$. Although this empirical function is not smooth, they use a stochastic gradient which acts as a smoothing operator, and they then use SGLD to find an approximate minimizer for the smoothed empirical function.

In our approach we instead obtain a smoothed version of $F$ by approximating the zero-one loss with a very steep logistic loss, and show that minimizing this smoothed function gives an approximate minimizer for the zero-one population loss function $\tilde{f}(x) := \frac{1}{r}\textstyle{\sum_{i=1}^r} \hat{\ell}(\lambda x ; (\mathcal{X}_i, \mathcal{Y}_i))$ 
for some scaling constant $\lambda >0$, where  $\hat{\ell}(a ; (s, b)) := b \varphi( \theta^\top a) - (1-b)\varphi(- \theta^\top a)$.

Towards this end, we consider the problem of optimizing the function $\mathsf{F}(x) :=  F(\frac{x}{\lambda})$ on the set $\mathsf{S} := \mathcal{T}^{\frac{1}{2}} \lambda [B \backslash  \frac{1}{2} B ]$, where $B$ is the unit ball.  To find an approximate global minimizer of $F$, we run the MALA Markov chain with stationary distribution $\propto e^{-U(x)}$, where $U(x) := \mathcal{T}^{-1} \tilde{f}(\frac{x}{d^{\frac{1}{4}}\lambda})$ at the inverse temperature $\mathcal{T}^{-1} =   \frac{c_1 d^{\frac{3}{2}}}{ \mathfrak{q}_0\epsilon^2}$, with $\lambda = \frac{100 \sqrt{d} }{\mathcal{T} \log(\mathcal{T})}$.  We show the following bound on the number of gradient evaluations required to find an $\epsilon$-approximate global minimizer for $F$:

\begin{theorem} \label{thm:ZeroOneLoss}
Suppose that $\epsilon < \frac{1}{10}$.  Let $U(x) := \mathcal{T}^{-1} \tilde{f}(\frac{x}{d^{\frac{1}{4}}\lambda})$. Then for any $\delta>0$ with probability at least $1-\delta$ Algorithm \ref{alg:optimization} generates a point $\hat{x}^\star$ such that $F(\hat{x}^\star) - \inf_{x\in \mathsf{S}}F(x) \leq \epsilon$ in $\mathcal{I} = \tilde{O}\left(d^{\frac{25}{6}} \mathfrak{q}_0^{\frac{11}{3}} \epsilon^{-\frac{22}{3}} \log(\frac{1}{\delta})  \log(\frac{\beta}{\delta}) \right)$ evaluations of $U$ and $\nabla U$.
\end{theorem}

\section{Technical overview}
\subsection{Proof for sampling}

To prove Theorem \ref{thm:main}, we use the conductance approach (see \cite{vempala2005geometric} for a survey): We first bound the conductance of the Markov chain in terms of the Cheeger constant, and then bound the mixing time in terms of the conductance. 

\paragraph{Bounding the conductance.}
To bound the conductance, we can use a result from \cite{lee2017convergence} (reproduced in our paper as Lemma \ref{lemma:conductance}) 
 which says that if for any $x,y$ with $\|x-y\|_2\leq \Delta$ we have $\|K(x,\cdot)- K(y,\cdot)\|_{\mathrm{TV}} \leq 0.97$,  then the conductance of the Markov chain with transition kernel $K$ is $\Psi_K = \Omega(\Delta \psi_\pi)$.  The bulk of our proof involves showing that if $K$ is the transition kernel of the MALA Markov chain with step size roughly\\ $\eta \leq \tilde{O}\left(\min \left(C_3^{-\frac{1}{3}}d^{-\frac{1}{6}} , d^{-\frac{1}{3}}, C_4^{-\frac{1}{4}}\right) \min(1, M^{-\frac{1}{2}})  [\log \log(\frac{1}{\mathsf{a}})]^{-1}\right)$, then $\|K(x,\cdot)- K(y,\cdot)\|_{\mathrm{TV}} \leq 0.97$ is satisfied whenever $\|x-y\|_2\leq \Delta$ for $\Delta = \frac{1}{100}(\frac{1}{2}\eta^{-1} + \frac{1}{4} \eta M)^{-1}$. 
 
 There are two steps in showing that the conditions of Lemma \ref{lemma:conductance} are satisfied for our choice of $\eta$: We first show that if the step size $\eta$ is small enough that the acceptance probability is at least $0.99$, then we have $\|K(x,\cdot)- K(y,\cdot)\|_{\mathrm{TV}} \leq 0.97$ whenever $\|x-y\|_2\leq \Delta$ for $\Delta = \frac{1}{100}(\frac{1}{2}\eta^{-1} + \frac{1}{4} \eta M)^{-1}$ (Lemma \ref{lemma:conductance1}).   We then show that, for our choice of $\eta$, the proposals made by the MALA algorithm have a $0.99$ acceptance probability whenever the position of the MALA Markov chain $X_i$ and the velocity term $V_i$ stay inside a certain ``good set" $G$ containing most of the probability measure of $\pi$.

\paragraph{Bounding the acceptance probability using Hamiltonian dynamics.} To bound the acceptance probability, we consider each step of the MALA Markov chain as an approximation to the trajectory of a particle in classical mechanics.  The MALA Markov chain proposes a step $\hat{X}_{i+1} = X_i+\eta V_i -\frac{1}{2} \eta^2  \nabla U(X_i)$, where $V_i \sim N(0,I_d)$.  This  proposal approximates the trajectory of a particle with initial position $X_i$ and initial velocity $V_i$. The total energy of this particle is given by the Hamiltonian functional $\mathcal{H}(x,v) := U(x) + \mathcal{K}(v)$, where $U(x)$ is the ``potential energy" of the particle and $\mathcal{K}(v) = \frac{1}{2}\|v\|^2$ is its ``kinetic energy".  The Hamiltonian is conserved for the continuous dynamics of this particle.

  Each step in MALA can be thought of as originating from one iteration of the leapfrog integrator, which approximates the position and velocity of this particle after time $\eta$
\be \label{eq:leapfrog}
\hat{X}_{i+1} &= X_i+\eta V_i -\frac{1}{2} \eta^2  \nabla U(X_i)\\
\hat{V}_{i+1} &= V_i - \eta \nabla U(X_i)  - \frac{1}{2}\eta^2 \frac{\nabla U(\hat{X}_{i+1})  - \nabla U(X_i)}{\eta} \approx   V_i - \eta \nabla U(X_i)  - \frac{1}{2}\eta^2 \nabla^2U(X_i) V_i.
\ee
 
Every time a proposal is made, the algorithm accepts the proposal with probability\\ $\min(e^{\mathcal{H}(\hat{X}_{i+1},\hat{V}_{i+1})- \mathcal{H}(X_i,V_i)},1)$.   
 If the proposal is accepted, the next step in the Markov chain is given by the position component $\hat{X}_{i+1}$ of the proposal.  The velocity component $\hat{V}_{i+1}$ is discarded after the accept-reject step and serves only to compute the acceptance probability.  To bound the acceptance probability $e^{\mathcal{H}(\hat{X}_{i+1},\hat{V}_{i+1})- \mathcal{H}(X_i,V_i)}$ we would like to bound the error $\mathcal{H}(\hat{X}_{i+1},\hat{V}_i)- \mathcal{H}(X_i,V_i)$ in the energy conservation for one step of the leapfrog integrator.  To do so, we use the fact that the continuous Hamiltonian dynamics conserves the Hamiltonian exactly.  Let $\hat{x}, \hat{v}$ be the position and velocity of the particle with continuous Hamiltonian dynamics after time $\eta$.  That is, $(\hat{x}, \hat{v}) = (q_\eta, p_\eta)$ are the solutions to Hamilton's equations 
\be
\frac{\mathrm{d}q_t}{\mathrm{d}t} =
  v_t \qquad \textrm{and} \qquad \frac{\mathrm{d}p_t}{\mathrm{d}t} 
     = -{\nabla U}(q_t),
\ee
evaluated at $t=\eta$, with initial conditions $(q_0, p_0) = (X_i, V_i)$. Since Hamilton's equations conserve the Hamiltonian, we have $\mathcal{H}(\hat{X}_{i+1},\hat{V}_i)- \mathcal{H}(X_i,V_i) =  \mathcal{H}(\hat{X}_{i+1},\hat{V}_i))- \mathcal{H}(\hat{x}, \hat{v})$.

To bound $\mathcal{H}(\hat{X}_{i+1},\hat{V}_i)- \mathcal{H}(\hat{x}, \hat{v})$, we separately bound the error $\hat{X}_{i+1}-X_i$ in the position and the error $\hat{V}_i-\hat{v}$ in the velocity.  To get a tight bound on these terms, we cannot simply bound their Euclidean norms, since the error in the Hamiltonian $\mathcal{H}(\hat{X}_{i+1},\hat{V}_i)$ is much larger when the position and momentum errors point in the worst-case direction where the Hamiltonian changes most quickly, than in a typical random direction (the worst-case direction is roughly $\nabla U(\hat{X}_{i+1})$ for the position error and $\hat{V}_i$ for the momentum error, since the gradient of the Hamiltonian is $\nabla \mathcal{H}(\hat{X}_{i+1},\hat{V}_i) = [\nabla U(\hat{X}_{i+1}); \hat{V}_i] $).

\paragraph{Bounding the kinetic energy error.}
We start by describing how to bound the kinetic energy error, since that is the most difficult task (Lemma \ref{lemma:kinetic}).   Since $\nabla \mathcal{K}(\hat{v}) = \hat{v}$, we have

\be \label{eq:overview1}
|\mathcal{K}(\hat{v}) - \mathcal{K}(\hat{V}_{i+1})| & \approx  |(\hat{V}_{i+1}- \hat{v})^\top \nabla \mathcal{K}(\hat{v})| = |(\hat{V}_{i+1}- \hat{v})^\top \hat{v}|\\
&\approx  |\int_0^\eta \int_0^r  V_i^\top [\nabla^2U(X_i)- \nabla^2U(X_i+V_i\tau)] V_i\mathrm{d}\tau \mathrm{d}r|,
\ee
where the last step is due to our approximation for $\hat{V}_{i+1}$ in terms of the Hessian-vector product $\nabla^2U(X_i) V_i$, and the fact that $\frac{\mathrm{d}^2p_t}{\mathrm{d}t^2} = \nabla^2U(q_t) \approx \nabla^2U(X_i+V_i t) $. (Equation \eqref{eq:leapfrog}).

Next, we bound the quantity in the integrand:
\be \label{eq:overview2}
|V_i^\top [\nabla^2U(X_i)&- \nabla^2U(X_i+V_i\tau)] V_i|\\
&=\left|\tau \nabla^3U(X_i)[V_i,V_i, V_i] + \tau \int_0^ \tau \nabla^3U(X_i+s V_i)[V_i,V_i, V_i] - \nabla^3U(X_i)[V_i,V_i, V_i] \mathrm{d} s \right|\\
& \approx \left|\tau \nabla^3U(X_i)[V_i,V_i, V_i] + \tau^2 \nabla^4U(X_i)[V_i,V_i,V_i,V_i] \right|\\
& \leq \tau C_3 \|\mathsf{X}^\top V_i\|_{\infty}^2  \|\mathsf{X}^\top V_i\|_2 + \tau^2 C_4  \|\mathsf{X}^\top V_i\|_{\infty}^4,
\ee
where the inequality holds by Assumption \ref{assumption:derivatives}. Combining Inequalitites \eqref{eq:overview1} and \eqref{eq:overview2}, we have
\be \label{eq:overview3}
|\mathcal{K}(\hat{v}) - \mathcal{K}(\hat{V}_{i+1})| \leq \eta^3 C_3 \|\mathsf{X}^\top V_i\|_{\infty}^2  \|\mathsf{X}^\top V_i\|_2 + \eta^4 C_4  \|\mathsf{X}^\top V_i\|_{\infty}^4
\ee

\noindent
 We show that the Kinetic energy error is $O(1)$ as long as the Markov chain $X_i$ and the velocity variable $V_i$ stay inside the ``good set" $G$. Roughly, we define $G$ to be the subset of $\mathbb{R}^{2d}$ where $\|\mathsf{X}^\top V_i\|_{\infty} \leq O(\log(\frac{d}{\delta}))$, $\|V_i\|_2 \leq O(\sqrt{d} \log(\frac{1}{\delta}))$, and $\|X_i - x^\star\|_2 \leq O(\frac{\sqrt{d}}{M} \log(\frac{1}{\delta}))$. Thus, whenever $(X_i,V_i)$ are in the good set, the first term on the right-hand side of Inequality \eqref{eq:overview3} is $O(1)$ if roughly $\eta \leq \tilde{O}(C_3^{-\frac{1}{3}}d^{-\frac{1}{6}} M^{-\frac{1}{2}} \log^{-1}(\frac{d}{\delta}))$.  The second term is $O(1)$ if  $\eta \leq O(C_4^{-\frac{1}{4}})$.

\paragraph{Bounding the potential energy error.}
To bound the potential energy error (Lemma \ref{lemma:potential}), we observe that $\hat{X}_{i+1} - \hat{x} \approx \int_0^\eta \int_0^t \nabla U(q_\tau)  - \nabla U(X_i) \mathrm{d}\tau \mathrm{d}t$ and hence that
\be
|U(\hat{X}_{i+1}) - U(\hat{x})| 
&\approx \left | \int_0^\eta \int_0^t  [\nabla U(q_\tau)  - \nabla U(X_i)]^\top \nabla U(X_i) \mathrm{d}\tau \mathrm{d}t \right |\\
&\approx \left | \int_0^\eta \int_0^t  [\nabla U(X_i + \tau V_i)  - \nabla U(X_i)]^\top \nabla U(X_i) \mathrm{d}\tau \mathrm{d}t \right |
\approx \left | \eta^2 [\nabla^2 U(X_i) \eta V_i]^\top \nabla U(X_i) \right |\\
&\leq \left | \eta^3  M^2 \|X_i\|_2 g_1 \right|,
\ee
for some $g_1 \sim N(0,1)$.  Hence, if we choose $\eta \leq d^{\frac{1}{3}} \min(1, M^{-\frac{1}{2}}) \log(\frac{1}{\delta})$ the potential energy error is $O(1)$ with probability at least $1-\delta$.

\paragraph{Bounding the probability of escaping the ``good set".} Finally, we show that, since our Markov chain is given a warm start, and $\pi$ has exponential tails, the Markov chain $X_i$ stays inside the good set $G$ with probability at least $1-\delta$ (Lemmas \ref{lemma:good_set} and \ref{lemma:WarmPreserved}).
\subsection{Proof for optimization}
The proof for optimization is similar to the proof for sampling, except that we bound the restricted Cheeger constant and restricted conductance, which were originally introduced in \cite{zhang2017hitting}, in place of the usual Cheeger constant and conductance.  We then apply a result from \cite{zhang2017hitting} (reproduced here as Lemma \ref{lemma:hitting}) to bound the hitting time to the set $\mathcal{U}$ as a function of the restricted conductance $\hat{\Psi}_K(\mathsf{S}\backslash \mathcal{U})$.

 The acceptance probability is bounded in the same way as in the proof for sampling, using the same choice of step size $\eta$.  The main difference is that we prove an analogue of Lemma \ref{lemma:conductance} which allows us to bound the restricted conductance in terms of the restricted Cheeger constant.  Specifically we show that if for any $x,y \in \mathsf{S}$ with $\|x-y\|_2\leq \Delta$ we have $\|K(x,\cdot)- K(y,\cdot)\|_{\mathrm{TV}} \leq 0.99$, then the restricted conductance of our Markov chain is $\hat{\Psi}_K(V) = \Omega(\Delta \hat{\psi}_{\pi}(V_\Delta))$  (Lemma \ref{lemma:restricted_conductance}).

\section{Defining the ``good set" and warm start.}

\begin{definition} [$\beta \geq 0$]
We say that $X_0 \sim  \mu_0$ is a $\beta$-warm start if
\be
\sup_{A \subset \mathbb{R}^d}\left(\frac{\mu_0(A)}{\pi(A)}\right) \leq \beta.
\ee
\end{definition}

In this case, there exists an event $E$ with $\pi(E) \geq \frac{1}{\beta}$ such that $\mu_0 = \pi | E$.

\begin{definition}
For $\alpha >\sqrt{2},R>0$, define the ``good set" $G$ as follows:
\be
G = \left\{(x,v)\in \mathbb{R}^d \textrm{ s.t. }\|\mathsf{X}^\top p_t(x,v)\|_\infty \leq \alpha \textrm{ for }t \in [0,T], \|q_t(x,v) - x^\star\|_2 \leq  \frac{3}{\sqrt{2}} \frac{R}{\sqrt{M}}, \|v\| \leq R \right\}.
\ee
\end{definition}

We set the step size as follows:
\be
\eta \leq O\left(\min \left(C_3^{-\frac{1}{3}}R^{-\frac{1}{3}} , R^{-\frac{2}{3}}, C_4^{-\frac{1}{4}}\right) \min(M^{-\frac{1}{2}}, 1) \alpha^{-1}\right),
\ee
where $\alpha>0$ will be fixed later in Section \ref{sec:proof_main}.

\section{Bounding conductance in terms of Cheeger constants} \label{sec:Conductance_bounds}
We recall the following bound for the conductance:
\begin{lemma}  [Lemma 13 in \cite{lee2017convergence}] \label{lemma:conductance}
Let $X$ be a time-reversible Markov chain with transition kernel $K$ and stationary distribution $\pi$.  Suppose that for any $x,y$ with $\|x-y\|_2\leq \Delta$ we have $\|K(x,\cdot)- K(y,\cdot)\|_{\mathrm{TV}} \leq 0.9$.  Then the conductance of $X$ is $\Psi_K = \Omega(\Delta \psi_\pi)$.
\end{lemma}

Next, we show a related bound on the restricted conductance:
\begin{lemma} [Restricted conductance] \label{lemma:restricted_conductance}
Let $\pi$ be a probability distribution on $\mathsf{S} \subseteq \mathbb{R}^d$, let $V \subseteq \mathsf{S}$, and let $X$ be a time-reversible Markov chain with transition Kernel $K$ and stationary distribution $\pi$.  Suppose that for any $x,y \in \mathsf{S}$ with $\|x-y\|_2\leq \Delta$ we have $\|K(x,\cdot)- K(y,\cdot)\|_{\mathrm{TV}} \leq 0.99$.  Then the restricted conductance of $X$ is $\hat{\Psi}_\pi(V) = \Omega(\Delta \hat{\psi}_{\pi}(V_\Delta))$.
\end{lemma}

\begin{proof}
Let $\rho_x = K(x,\cdot)$ be the transition distribution at $x$.  For any $S \subseteq \mathsf{S}$, let
\be
S^{(1)} &=\{x\in S : \rho_x(\mathsf{S} \backslash S) < 0.05 \}\\
S^{(2)} &=\{x\in \mathsf{S} \backslash S : \rho_x(S) < 0.05 \}\\
S^{(3)} &=\mathsf{S} \backslash (S^{(1)} \cup S^{(2)}).
\ee

Then the Euclidean distance between $S_1$ and $S_2$ is at least $\Delta$.

Without loss of generality, we may assume that $\pi(S_1) \geq \frac{1}{2} \pi(S)$, since otherwise we would have $\int_S \rho_x(\mathsf{S} \backslash S) \mathrm{d} \pi(x) = \Omega(1)$, implying a conductance of $\Omega(1)$.

\be
\pi(S^{(3)}) \geq \pi(S^{(1)}_\Delta) -  \pi(S^{(1)}) \geq   \Delta \times \hat{\psi}_\pi(V_\Delta) \times \pi(S^{(1)}).
\ee
We can now bound the restricted conductance:
\be
\int_S \rho_x(\mathsf{S} \backslash S) \mathrm{d} \pi(x) &\stackrel{{\scriptsize \textrm{Reversibility} }}{=}  \frac{1}{2}\left( \int_S \rho_x(\mathsf{S} \backslash S) \mathrm{d} \pi(x) + \int_{\mathsf{S} \backslash S} \rho_x(S) \mathrm{d} \pi(x) \right)\\
&\geq \frac{1}{2} \int_{S^{(3)}} 0.05 \mathrm{d} \pi(x)\\
&= 0.025 \pi(S^{(3)})\\
&\geq 0.025 \Delta \times \hat{\psi}_\pi(V_\Delta) \times \pi(S^{(1)})\\
&\geq 0.0125 \Delta \times \hat{\psi}_\pi(V_\Delta) \times \pi(S).
\ee

Hence, we have
\be
\hat{\Psi}_K(S) = \inf_{S\subseteq \mathsf{S}} \frac{\int_S \rho_x(\mathsf{S} \backslash S) \mathrm{d} \pi(x)}{\pi(S)} \geq 0.0125 \Delta \times \hat{\psi}_\pi(V_\Delta).
\ee
\end{proof}

\section{Bounding the mixing and hitting times as a function of conductance}

\begin{lemma} [Theorem 1.4 in \cite{lovasz1993random}] \label{lemma:mixing}
Let $X$ be a Markov chain with transition kernel $K$ and stationary distribution $\pi$ and initial distribution $\mu_0$. Suppose that $X$ is given a $\beta$-warm start (that is, $\mu_0(x) \leq \beta \pi(x)$ for every $x \in \mathbb{R}^d$).  Then for any $\hat{\epsilon}>0$ we have
\be
\|\mathcal{L}(X_i) - \pi\|_{\mathrm{TV}} \leq \hat{\epsilon}+\sqrt{\frac{\beta}{\hat{\epsilon}}} \left(1-\frac{1}{4} \Psi_K^2 \right)^i \qquad \forall i \in \mathbb{N}.
\ee
\end{lemma}

\begin{lemma} [Lemma 11 in \cite{zhang2017hitting}] \label{lemma:hitting}
Let $X$ be a time-reversible lazy Markov chain on $\mathsf{S}\subseteq \mathbb{R}^d$ with stationary distribution $\pi$ with initial distribution $\mu_0$. Let $\mathcal{U} \subseteq \mathsf{S}$. Suppose that $X$ is given a $\beta$-warm start on $\mathsf{S} \backslash \mathcal{U}$ (that is, $\mu_0(x) \leq \beta \pi(x)$ for every $x \in \mathsf{S} \backslash \mathcal{U}$).  Then for any $\delta>0$, the hitting time of $X$ to the set $\mathcal{U}$ is
\be
\inf\{i :X_i \in \mathcal{U}\} \leq \frac{4 \log(\frac{\beta}{\delta})}{\hat{\Psi}_K^2(\mathsf{S}\backslash \mathcal{U})},
\ee
with probability at least $1-\delta$.
\end{lemma}

\section{Exit probability from good set}
\begin{lemma} \label{lemma:good_set}
Let $x\sim \pi$, $v\sim N(0,I_d)$.  Then $P((x,v)\in G) \geq 1- N r e^{-\frac{16\alpha^2 -1}{8}} -e^{-\frac{R^2-d}{8}} - N e^{-\frac{\mathsf{a}}{\sqrt{d}} \frac{R}{\sqrt{M}}}$,
where  $N = 50\lceil (R+1) M^\frac{1}{2} \eta \rceil$.
\end{lemma}
\begin{proof}

Let $\mathcal{I} := \{\frac{\eta}{N}, 2 \frac{\eta}{N}, \ldots,  N \frac{\eta}{N} \}$, where $N = \lceil R M^\frac{1}{2} \eta \rceil$.  Then $p_t(x,v) \sim N(0,I_d)$ for all $t \in \mathcal{I}$.  Therefore by the Hanson-wright inequality we have that
\be
\mathbb P(\|\mathsf{X}^\top p_t(x,v)\|_\infty \leq \gamma) \leq r e^{-\frac{\gamma^2 -1}{8}} \quad \quad  \textrm{ for } \gamma> \sqrt{2}.
\ee
and hence that
\be \label{eq:invariance1}
\mathbb P(\max_{t\in \mathcal{I}} \|\mathsf{X}^\top p_t(x,v)\|_\infty \leq \gamma) \leq N r e^{-\frac{\gamma^2 -1}{8}} \quad \quad  \textrm{ for } \gamma> \sqrt{2}.
\ee

By the Hanson-Wright inequality, $\mathbb{P}[\|v\|>\xi]  \leq e^{-\frac{\xi^2-d}{8}} \quad \quad  \textrm{ for } \xi> \sqrt{2d}$.

Suppose that $\|q_t(x,v) - x^\star\|_2 \leq \frac{R}{\sqrt{M}}$ for all $t\in \mathcal{I}$ (by Assumption \ref{assumption:tails}, this occurs w.p. at least $1-Ne^{-\frac{\mathsf{a}}{\sqrt{d}} \frac{R}{\sqrt{M}}}$).

Then $\mathcal{H}(x,v) = U(x) + \frac{1}{2}\|v\|_2^2 \leq \mathrm{pi}^2 M\|x -x^\star\|_2^2  + \frac{1}{2}R^2  \leq 11 R^2$
Thus, $\|p_t(x,v)\|_2 \leq \sqrt{22} R$ for all $t \in \mathbb{R}$.  Thus,
\be
\|q_t(x,v) - x^\star\|_2 \leq \frac{R}{\sqrt{M}} + \frac{\eta}{N} \sqrt{22} R  \leq 2\frac{R}{\sqrt{M}}  \qquad \forall t \in \mathbb{R}.
\ee

Therefore, by the conservation of the Hamiltonian, with probability at least $1-e^{-\frac{R^2-d}{8}} -Ne^{-\frac{\mathsf{a}}{\sqrt{d}} \frac{R}{\sqrt{M}}}$,  for all $t\in \mathbb{R}$  we have $\|q_t(x,v) - x^\star\|_2 \leq 2 \frac{R}{\sqrt{M}}$, and hence that $\|\nabla U(q_t)\|_2 \leq 2M \frac{R}{\sqrt{M}}$.

Thus, since $\|p_{t+\frac{1}{M^\frac{1}{2}\sqrt{2d}}} - p_t\|_2 \leq \|\nabla U(q_t)\|_2 \times \frac{1}{RM^\frac{1}{2}} \leq 2 M  \frac{R}{\sqrt{M}} \times \frac{1}{R M^\frac{1}{2}} \leq 2$, by equation \eqref{eq:invariance1} we have
\be \label{eq:invariance2}
\mathbb P(\max_{t\in [0,\eta]} \|\mathsf{X}^\top p_t(x,v)\|_\infty \leq \gamma +2) \leq N r e^{-\frac{\gamma^2 -1}{8}} \quad \quad  \textrm{ for } \gamma > \sqrt{2}.
\ee

Thus, $\mathbb P((x,v) \in G) \geq 1- N r e^{-\frac{16 \alpha^2 -1}{8}} -e^{-\frac{R^2-d}{8}} - e^{-\frac{\mathsf{a}}{\sqrt{d}} \frac{R}{\sqrt{M}}}$.

\end{proof}

\section{Conductance bounds}

Let $\hat{a}_{z,v} : \mathbb{R}^d \rightarrow [0,1]$, and let $a_z = \mathbb{E}_v{\sim N(0,I_d)}[a_{z,v}]$.  Let $V_0,V_1, \ldots \sim N(0,I_d)$ i.i.d. and consider the following Markov chain:

\be
\mathsf{Z}_{i+1} = \begin{cases}
\mathsf{Z}_i+ \eta V_i - \frac{1}{2} \eta^2 \nabla U(\mathsf{Z}_i) \qquad \textrm{ with probability } \hat{a}_{\mathsf{Z}_i,V_i}\\
\mathsf{Z}_i \qquad \textrm{otherwise}.
\end{cases}
\ee
and let $K_{\mathsf{Z}}$ denote the probability transition Kernel of $\mathsf{Z}$.
  Let $\rho_z$ be the probability distribution of the next point in this Markov chain given that the current point is $z \in \mathbb{R}^d$, that is, $\rho_z = K_{\mathsf{Z}}(z,\cdot)$.

\begin{lemma} \label{lemma:conductance1}
   Suppose that for some $\eta>0$ and $x,y \in \mathbb{R}^d$ we have $a_x, a_y \geq 0.99$ and $\|x-y\|_2 \leq  \frac{1}{100}(\frac{1}{2}\eta^{-1} + \frac{1}{4} \eta M)^{-1}$.  Then we have $\|\rho_x - \rho_y\|_{\mathrm{TV}} < \frac{3}{100}$.
\end{lemma}

\begin{proof}

For any $z\in \mathbb{R}^d$, let $\gamma_z := z+ \eta v - \frac{1}{2} \eta^2 \nabla U(x)$ where $v \sim N(0,I_d)$.

Then $\gamma_z \sim N(z - \frac{1}{2} \eta^2 \nabla U(z), \eta^2 I_d)$.

Therefore, by Theorem 1.3 in \cite{devroye2018total}, we have
\be
\|\mathcal{L}(\gamma_x) - \mathcal{L}(\gamma_y)\|_{\mathrm{TV}} &\leq \frac{\|x-y - \frac{1}{2} \eta^2( \nabla U(x) - \nabla U(y))\|_2}{2 \eta}\\
&\leq \frac{\|x-y\|_2 + \frac{1}{2} \eta^2\| \nabla U(x) - \nabla U(y)\|_2}{2 \eta}\\
&\leq \frac{\|x-y\|_2 + \frac{1}{2} \eta^2 M\| x - y\|_2}{2 \eta}\\
&=(\frac{1}{2}\eta^{-1} + \frac{1}{4} \eta M)\|x-y\|_2.
\ee

Hence, since $\|x-y\|_2 \leq  \frac{1}{100}(\frac{1}{2}\eta^{-1} + \frac{1}{4} \eta M)^{-1}$, we have
\be
\|\mathcal{L}(\gamma_x) - \mathcal{L}(\gamma_y)\|_{\mathrm{TV}} \leq \frac{1}{100}.
\ee

Thus, since $a_x, a_y \geq 0.99$, we have
\be
\|\rho_x - \rho_y\|_{\mathrm{TV}} \leq \frac{1}{100} +\frac{2}{100}< \frac{3}{100}.
\ee
\end{proof}

\begin{lemma} \label{lemma:conductance2}
Let $\pi$ be the distribution $\pi(x) \propto e^{-U(x)}$.  Suppose that for some $\eta>0$ and any $x,y \in \mathbb{R}^d$ the acceptance probability from both $x$ and $y$ is $a_x, a_y\geq 0.97$.  Then the conductance $\Psi_{K_{\mathsf{Z}}}$ is $\Omega((\frac{1}{2}\eta^{-1} + \frac{1}{4} \eta M)^{-1} \psi_\pi)$.
\end{lemma}

\begin{proof}
This follows by applying Lemma \ref{lemma:conductance1} to Lemma \ref{lemma:conductance}.
\end{proof}

Now consider the Markov chain $\mathsf{\hat{Z}}$ defined by the recursion
\be
\mathsf{\tilde{Z}}_{i+1} &= \begin{cases}
\mathsf{\hat{Z}}_i+ \eta V_i - \frac{1}{2} \eta^2 \nabla U(\mathsf{\hat{Z}}_i) \qquad \textrm{ with probability } \hat{a}_{\mathsf{\hat{Z}}_i,V_i}\\
\mathsf{\hat{Z}}_i \qquad \textrm{otherwise}.\\
\end{cases}\\
\mathsf{\hat{Z}}_{i+1} &= \begin{cases}
\tilde{Z}_i \qquad \textrm{ if } \tilde{Z}_i \in \mathsf{S}\\
\mathsf{\hat{Z}}_i \qquad \textrm{otherwise}.
\end{cases}
\ee
and let $K_{\mathsf{\tilde{Z}}}$ denote the probability transition Kernel of $\mathsf{\tilde{Z}}$.
\begin{lemma} \label{lemma:restricted_conductance2}
Let $\pi$ be the distribution $\pi(x) \propto e^{-U(x)} \times \mathbbm{1}_\mathsf{S}(x)$.   Suppose that for some $\eta>0$ and any $x,y \in \mathbb{R}^d$ that $a_x, a_y \geq 0.99$.  Let $v\sim N(0,I_d)$, and suppose that $x+ \eta v - \frac{1}{2} \eta^2 \nabla U(x) \in \mathsf{S}$ with probability at least $\frac{1}{10}$.  Then for any subset $V \subseteq \mathsf{S}$, the restricted conductance is $\hat{\Psi}_{K_{\mathsf{\tilde{Z}}}}(V) = \Omega(\Delta \hat{\psi}_{\pi}(V_\Delta))$, where $\Delta = \frac{1}{100}(\frac{1}{2}\eta^{-1} + \frac{1}{4} \eta M)^{-1}$.
\end{lemma}
\begin{proof}
First, we note that for $v_1, v_2 \sim N(0,I_d)$ we have $x+ \eta v_1 - \frac{1}{2} \eta^2 \nabla U(x) \in \mathsf{S}$ with probability at least $\frac{1}{10}$ and $y+ \eta v_2 - \frac{1}{2} \eta^2 \nabla U(y) \in \mathsf{S}$ with probability at least $\frac{1}{10}$.

By Lemma \ref{lemma:conductance1}, we have $\|\rho_x - \rho_y\|_{\mathrm{TV}} < \frac{3}{100}$ whenever $\|x-y\|_2 \leq \Delta$, where $\Delta = \frac{1}{100}(\frac{1}{2}\eta^{-1} + \frac{1}{4} \eta M)^{-1}$.

Hence, whenever $\|x-y\|_2 \leq \Delta$ we have
\be
\|K(x,\cdot)- K(y,\cdot)\|_{\mathrm{TV}} \leq   1- (\frac{1}{10} - \|\rho_x - \rho_y\|_{\mathrm{TV}}) \leq 1- \frac{7}{100} \leq 0.99.
 \ee

Thus by Lemma \ref{lemma:restricted_conductance}, we have that for any subset $V \subseteq \mathsf{S}$, the restricted conductance is $\hat{\Psi}_{K_{\mathsf{\tilde{Z}}}}(V) = \Omega(\Delta \hat{\psi}_{\pi}(V_\Delta))$
\end{proof}

\begin{lemma} \label{lemma:WarmPreserved}
Consider any Markov chain $Z$ on $\mathbb{R}^d$ and denote by $K(\cdot, \cdot)$ its transition kernel and its stationary distribution by $\pi$.  Suppose that $K$ satisfies the detailed balance equations, that is, $\pi(x) K(x,y) = \pi(y) K(y,x)$ for all  $x,y \in \mathbb{R}^d$.  Then
for every  $k \in \mathbb{Z}^\star$,
\be
\sup_{A \subset \mathbb{R}^d, \pi(A)\neq 0}\left(\frac{\mu_k(A)}{\pi(A)}\right) \leq \beta.
\ee
\end{lemma}
\begin{proof}

We will prove this by induction.  Suppose (towards an induction) that for some $k\in \mathbb{Z}^\ast$ we have 
\be \label{eq:InductiveAssumption}
\frac{\mu_k(y)}{\pi(y)} \leq \beta \qquad \forall y \in \mathbb{R}^d \textrm{ s.t. } \pi(y) \neq 0.
\ee
Since we have a $\beta$-warm start, Inequality \eqref{eq:InductiveAssumption} is satisfied for $k= 0$.  Now we will show that if our inductive assumption \eqref{eq:InductiveAssumption} is satisfied for some $k \in \mathbb{Z}^\ast$, it is also satisfied for $k+1$.

The proof follows from the fact that the Markov chain satisfies the detailed balance equations:
\be \label{eq:DetailedBalance}
\pi(x) K(x,y) = \pi(y) K(y,x) \qquad \forall x,y \in \mathbb{R}^d.
\ee
Then
\be
\frac{\mu_{k+1}(x)}{\pi(x)} = \int_{\mathbb{R}^d} \frac{K(y,x)}{\pi(x)} \mu_{k}(y) \mathrm{d}y \stackrel{{\scriptsize \textrm{Eq.} \ref{eq:DetailedBalance}}}{=} \int_{\mathbb{R}^d} \frac{K(x,y)}{\pi(y)} \mu_{k}(y) \mathrm{d}y
\stackrel{{\scriptsize \textrm{Eq.} \ref{eq:InductiveAssumption}}}{\leq} \int_{\mathbb{R}^d}  K(x,y) \beta  \mathrm{d}y = \beta \int_{\mathbb{R}^d}  K(x,y)  \mathrm{d}y  = \beta.
\ee

\end{proof}

\section{Proof of main theorem for sampling} \label{sec:proof_main}

\begin{proof}[Proof of Theorem \ref{thm:main}]

Without loss of generality, we may assume that $U$ has a global minimizer $x^{\star}$ at $x^{\star} = 0$ (since we assume that the initial point $X_0$ has a $\beta$-warm start with respect to $U$ but do not assume anything about the location of $X_0$ with respect to the origin).

Set $\mathcal{I} = 10^4 ((\eta^{-1} + \eta L)\psi)^{-2} \log(\frac{\beta}{\epsilon})$.

Choose $\alpha = \log(\frac{\mathcal{I}\beta N}{\epsilon})$ and $R = \sqrt{d} \log(\frac{1}{\epsilon}\max(1,\frac{\sqrt{M}}{\mathsf{a} \mathcal{I} \beta N}))$, where $N = \lceil R M^\frac{1}{2} \eta \rceil$.

By Lemmas \ref{lemma:good_set} and \ref{lemma:WarmPreserved}, we have that,
\be
\mathbb P((X_i, V_i) \in G \, \, \forall \, \, 0 \leq i \leq \mathcal{I}-1) \geq 1- \mathcal{I} \times \beta \times [N r e^{-\frac{16 \alpha^2 -1}{8}} -e^{-\frac{R^2-d}{8}} - e^{-\frac{\mathsf{a}}{\sqrt{d}} \frac{R}{\sqrt{M}}}] \geq 1- \epsilon
\ee
Therefore, by Lemmas \ref{lemma:potential} and \ref{lemma:kinetic} with probability at least $1-\frac{\epsilon}{\mathcal{I}}$, the acceptance probability is
\be
\min(1, \, e^{\mathcal{H}(\hat{X}_i,\hat{V}_i)-\mathcal{H}(X_i, V_i)}) \geq e^{-\frac{2}{10}} > 0.8.
\ee

Let $i^\star = \min\{i : (X_i,V_i)\notin G \}$.  Then with probability at least $1- \mathcal{I}\times \frac{\epsilon}{\mathcal{I}} = 1-\epsilon$, we have that $\mathcal{I}\leq i^\star$.
 Consider the toy Markov chain $X^{\dagger}$, where
\be
X^{\dagger}_i = \begin{cases} X_i \qquad \textrm{ if } i < i^\star \\ Y_i \textrm{ if } i \geq i^\star \end{cases},
\ee
and where $Y_1, Y_2 \ldots \sim \pi$ are i.i.d. and independent of $X_0,\ldots, X_{i^\star-1}$.  Denote the transition kernel of $X^{\dagger}$ by $K^{\dagger}$.

Then by Lemma \ref{lemma:conductance2} we have that the conductance $\Psi_{K^{\dagger}}$ of the $X^{\dagger}$ chain is $\Omega((\frac{1}{2}\eta^{-1} + \frac{1}{4} \eta L)^{-1} \psi_\pi)$.

Then by Theorem 11 in \cite{lee2017convergence}, we have
\be
\|\mathcal{L}(X^{\dagger}_i) -\pi \|_{\mathrm{TV}} \leq \epsilon + \sqrt{\frac{1}{\epsilon} \beta}\left(1-\frac{1}{2}\Psi_{\textrm{$K^{\dagger}$}}^2\right)^i.
\ee
Hence,
\be
\|\mathcal{L}(X^{\dagger}_i) -\pi \|_{\mathrm{TV}} \leq 2\epsilon \qquad \forall i \geq \Omega \left (\Psi_{\textrm{$K^{\dagger}$}}^{-2} \log(\frac{\beta}{\epsilon})\right).
\ee
Therefore, since with probability at least $1-\epsilon$ we have $X_i = X^{\dagger}_i$, it must be that
\be
\|\mathcal{L}(X_{\mathcal{I}}) -\pi \|_{\mathrm{TV}} \leq 3\epsilon.
\ee

\end{proof}

\subsection{Bounding the potential energy error}
For every $t>0$, define
\be
\hat{q}_t &:= q_0 + t p_0  - \frac{1}{2}t^2 \nabla U(q_0)\\
\hat{p}_t &:= p_0 - t \nabla U(q_0)  - \frac{1}{2}t^2 \nabla^2U(q_0) p_0.
\ee
\begin{lemma}[potential energy error] \label{lemma:potential}
If $(X_i,V_i) \in G$, then with probability at least $1- \frac{\epsilon}{\mathcal{I}}$ we have $|U(\hat{X}_i)-U(X_i)| \leq \frac{1}{10}$.
\end{lemma}

\begin{proof}
First, we note that
\be
&q_t =  q_0 + t p_0  - \int_0^t \int_0^r \nabla U(q_r) \mathrm{d}\tau \mathrm{d}r  = q_0 + t p_0  - \left[\frac{1}{2}t^2 \nabla U(q_0) + \int_0^t \int_0^r \nabla U(q_\tau) - \nabla U(q_0)\mathrm{d} \tau \mathrm{d}r \right]\\
&\hat{q}_t = q_0 + t p_0  - \frac{1}{2}t^2 \nabla U(q_0) \qquad \forall t>0.
\ee

Thus,
\be
U(q_t) - U(\hat{q}_t) &= \int_0^1 (q_t- \hat{q}_t)^\top \nabla U(s (q_t - \hat{q}_t) + \hat{q}_t) \mathrm{d}s\\
&  = \int_0^1 (q_t- \hat{q}_t)^\top \nabla U(q_0)\mathrm{d}s + \int_0^1 (q_t- \hat{q}_t)^\top [\nabla U(s (q_t - \hat{q}_t) + \hat{q}_t) - \nabla U(q_0)] \mathrm{d}s\\
&  =- \left(\int_0^t \int_0^r  \nabla U(q_\tau)  - \nabla U(q_0) \mathrm{d}\tau \mathrm{d}r \right)^\top \nabla U(q_0)\\
&\qquad + \int_0^1 \left(\int_0^t \int_0^r \nabla U(q_\tau)  - \nabla U(q_0) \mathrm{d}\tau \right)^\top [\nabla U(s (q_t - \hat{q}_t) + \hat{q}_t) - \nabla U(q_0)] \mathrm{d}s\\
&  = -\int_0^t \int_0^r  \underbrace{[\nabla U(q_\tau)  - \nabla U(q_0)]^\top \nabla U(q_0)}_{\text{(1)}} \mathrm{d}\tau \mathrm{d}r\\
&\qquad + \int_0^1 \int_0^t \int_0^r \underbrace{[\nabla U(q_\tau)  - \nabla U(q_0)]^\top [\nabla U(s (q_t - \hat{q}_t) + \hat{q}_t) - \nabla U(q_0)]}_{\text{(2)}} \mathrm{d}\tau \mathrm{d}r \mathrm{d}s.\\
\ee

We start by bounding term (1):  
\be
|(1)| = \left |\nabla U(q_0)^\top[\nabla U(q_\tau) - \nabla U(q_0)] \right| &= \left |\nabla U(q_0)^\top \left [\nabla^2 U(q_0) \tau p_0 + \tau \int_0^\tau \left(\nabla^2U(q_s) - \nabla^2 U(q_0)\right) p_0 \mathrm{d} s \right] \right|\\
&\leq \tau M \|\nabla U(q_0) \| |g_1| + \tau^2  \|\nabla U(q_0) \|_2 \times \tau \sup_{0\leq s \leq \tau} \|p_s \|_2 \times {C_3} \|g\|_2\\
&\leq \tau M^2 \|q_0\|_2 |g_1| + \tau^2  M \|q_0\|_2 \times \tau \sup_{0\leq s \leq \tau} \|p_s \|_2 \times {C_3} \|g\|_2.
\ee
for some $g\sim N(0,I_d)$, since the random vector $p_0$ is probabilistically independent of the row-vector $\nabla U(q_0)^\top \nabla^2 U(q_0)$.  
 
 Next, we bound term (2):
\be
|(2)| = {[}\nabla U(q_\tau)  -& \nabla U(q_0)]^\top [\nabla U(s (q_t - \hat{q}_t) + \hat{q}_t) - \nabla U(q_0)]\\
&= [\nabla U(q_\tau)  - \nabla U(q_0)]^\top [(\nabla U(q_t) -  \nabla U(q_0))  +    (\nabla U(s (q_t - \hat{q}_t) + \hat{q}_t)  - \nabla U(q_t))]\\
&\leq M \|q_t - q_0\| \times M(\|q_t - q_0\| + \|q_t - \hat{q}_t\|)\\
&\leq M \|q_t - q_0\| \times M \left(\|q_t - q_0\| + \int_0^t \|\nabla U(q_\tau) - \nabla U(q_0)\| \mathrm{d}\tau \right)\\
&\leq M \|q_t - q_0\| \times M \left(\|q_t - q_0\| +  t \sup_{0\leq \tau \leq t}\|\nabla U(q_\tau) - \nabla U(q_0)\| \right)\\
&\leq M t \sup_{0\leq \tau \leq t} \|p_\tau \| \times M \left(t \sup_{0\leq \tau \leq t} \|p_\tau \| +  M t^2 \sup_{0\leq \tau \leq t} \|p_\tau \| \right).\\
\ee

Therefore,
\be
|U(q_t) - U(\hat{q}_t)| &\leq   t^3 M^2 \|q_0\|_2 |g_1| + t^4   M \|q_0\|_2 \times \tau \sup_{0\leq s \leq \tau} \|p_s \|_2 \times {C_3} \|g\|_2\\
& \qquad+ M t \sup_{0\leq \tau \leq t} \|p_\tau \| \times M \left(t \sup_{0\leq \tau \leq t} \|p_\tau \| +  M t^2 \sup_{0\leq \tau \leq t} \|p_\tau \| \right)\\
&\leq \frac{1}{100}.
\ee
with probability at least $1-\frac{\epsilon}{\mathcal{I}}$ whenever $(q_0, p_0) \in G$.
\end{proof}

\subsection{Bounding the kinetic energy error}
\begin{lemma}[kinetic energy error] \label{lemma:kinetic}
If $(X_i,V_i) \in G$, then with probability at least $1- \frac{\epsilon}{\mathcal{I}}$ we have $|\frac{1}{2}\|\hat{X}_i\|_2^2- \frac{1}{2}\|X_i\|_2^2| \leq \frac{1}{10}$.
\end{lemma}

\begin{proof}
Recall that $\mathcal{K}(p):= \frac{1}{2}\|p\|_2^2$ denotes the kinetic energy.  Then
\be \label{eq:KineticEnergy}
\mathcal{K}(p_t) - \mathcal{K}(\hat{p}_t) &= \int_0^1 (p_t- \hat{p}_t)^\top \nabla \mathcal{K}(s (p_t - \hat{p}_t) + \hat{p}_t) \mathrm{d}s\\
&= \int_0^1 (p_t- \hat{p}_t)^\top (s (p_t - \hat{p}_t) + \hat{p}_t) \mathrm{d}s\\
&=  (p_t- \hat{p}_t)^\top\hat{p}_t + \int_0^1 s\|p_t- \hat{p}_t\|^2 \mathrm{d}s\\
&=  (p_t- \hat{p}_t)^\top\hat{p}_t + \frac{1}{2}\|p_t- \hat{p}_t\|^2\\
&=  (p_t-  [p - t \nabla U(q_0)  - \frac{1}{2}t^2 \nabla^2U(q_0) p_0])^\top\hat{p}_t +[\frac{1}{2}t^2 \frac{\nabla U(\hat{q}_t)  - \nabla U(q_0)}{t} - \frac{1}{2}t^2 \nabla^2U(q_0) p_0]^\top\hat{p}_t + \frac{1}{2}\|p_t- \hat{p}_t\|^2\\
&=  (p_t-  [p_0 - t \nabla U(q_0)  - \frac{1}{2}t^2 \nabla^2U(q_0) p_0])^\top [p_0 - t \nabla U(q_0)  - \frac{1}{2}t^2 \nabla^2U(q_0) p_0]\\
&\qquad - (p_t-  [p_0 - t \nabla U(q_0)  - \frac{1}{2}t^2 \nabla^2U(q_0) p_0])^\top[\frac{1}{2}t^2 \frac{\nabla U(\hat{q}_t)  - \nabla U(q_0)}{t} - \frac{1}{2}t^2 \nabla^2U(q_0) p_0]\\
&\qquad \qquad +[\frac{1}{2}t^2 \frac{\nabla U(\hat{q}_t)  - \nabla U(q_0)}{t} - \frac{1}{2}t^2 \nabla^2U(q_0) p_0]^\top\hat{p}_t + \frac{1}{2}\|p_t- \hat{p}_t\|^2\\
&=  (\int_0^t \int_0^r [\nabla^2U(q_0)- \nabla^2U(q_\tau)] p_0 \mathrm{d}r \mathrm{d}\tau)^\top [p_0 - t \nabla U(q_0)  - \frac{1}{2}t^2 \nabla^2U(q_0) p_0]\\
&\qquad - (p_t-  [p_0 - t \nabla U(q_0)  - \frac{1}{2}t^2 \nabla^2U(q_0) p_0])^\top[\frac{1}{2}t^2 \frac{\nabla U(\hat{q}_t)  - \nabla U(q_0)}{t} - \frac{1}{2}t^2 \nabla^2U(q_0) p_0]\\
&\qquad \qquad +[\frac{1}{2}t^2 \frac{\nabla U(\hat{q}_t)  - \nabla U(q_0)}{t} - \frac{1}{2}t^2 \nabla^2U(q_0) p_0]^\top\hat{p}_t + \frac{1}{2}\|p_t- \hat{p}_t\|^2\\
&=  \left(\int_0^t \int_0^r  [\nabla^2U(q_0)- \nabla^2U(q_0+p_0\tau)] p_0 \mathrm{d}r \mathrm{d}\tau \right)^\top [p_0 - t \nabla U(q_0)  - \frac{1}{2}t^2 \nabla^2U(q_0) p_0]\\
&\qquad -(\int_0^t \int_0^r [(\nabla^2U(q_\tau)- \nabla^2U(q_0+p_0\tau))] p_0 \mathrm{d}r \mathrm{d}\tau)^\top [p_0 - t \nabla U(q_0)  - \frac{1}{2}t^2 \nabla^2U(q_0) p_0]\\
&\qquad \qquad - (p_t-  [p_0 - t \nabla U(q_0)  - \frac{1}{2}t^2 \nabla^2U(q_0) p_0])^\top[\frac{1}{2}t^2 \frac{\nabla U(\hat{q}_t)  - \nabla U(q_0)}{t} - \frac{1}{2}t^2 \nabla^2U(q_0) p_0]\\
&\qquad \qquad \qquad +[\frac{1}{2}t^2 \frac{\nabla U(\hat{q}_t)  - \nabla U(q_0)}{t} - \frac{1}{2}t^2 \nabla^2U(q_0) p_0]^\top\hat{p}_t + \frac{1}{2}\|p_t- \hat{p}_t\|^2\\
&=  \int_0^t \int_0^r  \underbrace{p_0^\top [\nabla^2U(q_0)- \nabla^2U(q_0+p_0\tau)] p_0}_{\text{(4)}}\mathrm{d}\tau \mathrm{d}r \\
&\qquad +  \left(\int_0^t \int_0^r [\underbrace{\nabla^2U(q_0)- \nabla^2U(q_0+p_0\tau)] p_0}_{\text{(5a)}} \mathrm{d}\tau \mathrm{d}r  \right)^\top \underbrace{[- t \nabla U(q_0)  - \frac{1}{2}t^2 \nabla^2U(q_0) p_0]}_{\text{(5b)}}\\
&\qquad \qquad -(\int_0^t \int_0^r \underbrace{[\nabla^2U(q_\tau)- \nabla^2U(q_0+p_0\tau)] p_0}_{\text{(6a)}} \mathrm{d}\tau \mathrm{d}r)^\top \underbrace{[p_0 - t \nabla U(q_0)  - \frac{1}{2}t^2 \nabla^2U(q_0) p_0]}_{\text{(6b)}}\\
&\qquad \qquad \qquad - (\underbrace{p_t-  [p_0 - t \nabla U(q_0)  - \frac{1}{2}t^2 \nabla^2U(q_0) p_0]}_{\text{(7a)}})^\top \underbrace{[\frac{1}{2}t^2 \frac{\nabla U(\hat{q}_t)  - \nabla U(q_0)}{t} - \frac{1}{2}t^2 \nabla^2U(q_0) p_0]}_{\text{(7b)}}\\
&\qquad \qquad \qquad \qquad +\underbrace{[\frac{1}{2}t^2 \frac{\nabla U(\hat{q}_t)  - \nabla U(q_0)}{t} - \frac{1}{2}t^2 \nabla^2U(q_0) p_0]^\top\hat{p}_t}_{\text{(8)}} + \underbrace{\frac{1}{2}\|p_t- \hat{p}_t\|_2^2}_{\text{(9)}}.\\
\ee

We now bound (1)-(9)

\begin{enumerate}

\item We start by bounding term (4): %
\be
|(4)| &= |p_0^\top [\nabla^2U(q_0)- \nabla^2U(q_0+\tau p_0)] p_0|\\
& = \left|\tau \nabla^3U(q_0)[p_0,p_0, p_0] + \tau \int_0^ \tau \nabla^3U(q_0+s p_0)[p_0,p_0, p_0] - \nabla^3U(q_0)[p_0,p_0, p_0] \mathrm{d} s \right|\\
& \leq \tau |\nabla^3U(q_0)[p_0,p_0, p_0]| + \tau^2 \mathbb{E}_{x\sim \mathrm{Unif}([q_0, q_0+s p_0])} \left|\nabla^4U(q_0)[p_0,p_0,p_0,p_0] \right|\\
& \leq \tau C_3 \|\mathsf{X}^\top p_0\|_{\infty}^2  \|\mathsf{X}^\top p_0\|_2 + \tau^2 C_4  \|\mathsf{X}^\top p_0\|_{\infty}^4
\ee
where $\mathrm{Unif}([q_0, q_0+s p_0])$ is the uniform distribution on the line segment connecting $q_0$ and $q_0+s p_0$.

\item Next, we bound term (5a). For any $v\in \mathbb{R}^d$ we have
\be
|v^\top (5a)| &= |z^\top [\nabla^2U(q_0)- \nabla^2U(q_0+p_0\tau)] p_0| = |\int_0^\tau \nabla^3 U(q_0+p_0 s)[p_0,p_0, v] \mathrm{d}s|\\
&\leq  \int_0^\tau |\nabla^3 U(q_0+p_0 s)[p_0,p_0, v]| \mathrm{d}s \leq  \tau C_3  \|\mathsf{X}^\top p_0\|_{\infty}^2 \| v\|_{2}.
\ee

\item Next, we bound term (5b)
\be
\|(5b)\|_{2} &= t \|\nabla U(q_0)\|_2  + \frac{1}{2}t^2 \|\nabla^2U(q_0) p_0 \|_2\\
& \leq t M \|q_0\|_2  + \frac{1}{2}t^2 M \|p_0 \|_2.
\ee

\item Next, we bound term (6a).  
First, observe that
\be \label{eq:6a1}
 \|q_\tau - (q_0+p_0)\tau\|_2  \leq \| \int_0^\tau \int_0^s \nabla U(q_r) \mathrm{d}r \mathrm{d}s \|_2 \leq  \tau^2 M \sup_{s\in[0,\tau]}\|q_s\|_2.
\ee

For any $v\in \mathbb{R}^d$ we have
\be
|v^\top (6a)| &= |v^\top [\nabla^2U(q_\tau)- \nabla^2U(q_0+p_0\tau)] p_0| = \int_0^1 \nabla^3 U\big((1-s)q_\tau + s(q_0+p_0\tau)\big)[p_0,p_0, v ] \mathrm{d} s\\
 &\leq C_3 \|q_\tau - (q_0+p_0)\tau\|_{2} \|\mathsf{X}^\top p_0\|_{\infty}^2 \|\mathsf{X}^\top v\|_{\infty}\\
 &\stackrel{{\scriptsize \textrm{Eq. }}\ref{eq:6a1}}{\leq} C_3 \tau^2 M \sup_{s\in[0,\tau]}\|q_s\|_2 \times \|\mathsf{X}^\top p_0\|_{\infty}^2 \|v\|_{2}.
\ee

\item Next, we bound term (6b)
\be
\|\mathsf{X}^\top (6b)\|_2 = \left \|\mathsf{X}^\top [p_0 - t \nabla U(q_0)  - \frac{1}{2}t^2 \nabla^2U(q_0) p_0] \right \|_2 \leq   \|p_0\|_2 + t \|q_0 \|_2 M  + \frac{1}{2}t^2  M\|p_0\|_2.
\ee

\item Next, we bound term (7a).
By the proof of Lemma 9.1 in the ArXiv version of \cite{mangoubi2018dimensionally_arxiv}, we have
\be \label{old_result}
\max \left (\|(7a)\|_2, \|(7b)\|_2, \sqrt{(9)}\right) \leq \frac{1}{3} t^3 \left[C_3 \sup_{t\in[0,\eta]}\|\mathsf{X}^\top p_t\|_{\infty}^{2} +  (M)^2 \sup_{t\in[0,\eta]} \|q_t\|_2 \right].\\
\ee
and finally, that $\|\hat{p}_t\|_2 \leq \|(6b)\|_2+ \|(7b)\|_2 \leq \|p_0\|_2 + t \|q_0 \|_2 M  + \frac{1}{2}t^2  M\|p_0\|_2 + \|(7b)\|_2$.

\item Next, we bound term (8)

First, we note that
\be \label{eq:momentum_diff}
 \|\hat{p}_t - p_0\|_2 &= \left \| t \nabla U(q_0)  - \frac{1}{2}t^2 \frac{\nabla U(q_0 + tp_0 - \frac{1}{2}t^2\nabla U(q_0))  - \nabla U(q_0)}{t} \right\|_2\\
 &\leq t\|\nabla U(q_0)\|_2 + \frac{1}{2}t^2 M \|p_0 - \frac{1}{2}t\nabla U(q_0)\|_2\\
 &\leq t\|\nabla U(q_0)\|_2 + \frac{1}{2}t^2 M \|p_0\|_2 + \frac{1}{2}t^3 M\|\nabla U(q_0)\|_2\\
 &\leq 2t\|\nabla U(q_0)\|_2 + \frac{1}{2}t^2 M \|p_0\|_2.\\
  &\leq 2t M \|q_0\|_2 + \frac{1}{2}t^2 M \|p_0\|_2.\\
\ee

Hence,
\be
(8) &= \left[\frac{1}{2}t^2 \frac{\nabla U(\hat{q}_t)  - \nabla U(q_0)}{t} - \frac{1}{2}t^2 \nabla^2U(q_0) p_0\right]^\top\hat{p}_t\\
& = \left[\frac{1}{2}t (\nabla U(q_0 + tp_0)  - \nabla U(q_0)) - \frac{1}{2}t^2 \nabla^2U(q_0) p_0\right]^\top\hat{p}_t + \frac{1}{2}t [\nabla U(q_0 + tp_0) - \nabla U(\hat{q}_t)]^\top\hat{p}_t\\
& = \left[\frac{1}{2}t (\nabla U(q_0 + tp_0)  - \nabla U(q_0)) - \frac{1}{2}t^2 \nabla^2U(q_0) p_0\right]^\top p_0 +  \left[\frac{1}{2}t (\nabla U(q_0 + tp_0)  - \nabla U(q_0)) - \frac{1}{2}t^2 \nabla^2U(q_0) p_0\right]^\top(\hat{p}_t - p_0)\\
&\qquad \qquad  + \frac{1}{2}t [\nabla U(q_0 + tp_0) - \nabla U(\hat{q}_t)]^\top p_0 + \frac{1}{2}t [\nabla U(q_0 + tp_0) - \nabla U(\hat{q}_t)]^\top(\hat{p}_t-p_0)\\
& = \frac{1}{2}t^2  \int_0^t \nabla^3U(q_0 + sp_0)[p_0,p_0, p_0] \mathrm{d}s +  \frac{1}{2}t^2  \int_0^t \nabla^3U(q_0 + sp_0)[p_0,p_0, \hat{p}_t - p_0] \mathrm{d}s\\
&\qquad \qquad  + \frac{1}{2}t [\nabla U(q_0 + tp_0) - \nabla U(q_0 + tp_0 - \frac{1}{2}t^2 \nabla U(q_0) )]^\top p_0 + \frac{1}{2}t [\nabla U(q_0 + tp_0) - \nabla U(\hat{q}_t)]^\top(\hat{p}_t-p_0)\\
&\leq \frac{1}{6}t^3 C_3 \|\mathsf{X}^\top p_0\|_\infty  \|p_0\|_2 + \frac{1}{6}t^3 C_3 \|\mathsf{X}^\top p_0\|_\infty^2  \|\hat{p}_t - p_0\|_2\\
&\qquad \qquad  + \frac{1}{2}t \left[\int_0^1 [\frac{1}{2}t^2 \nabla U(q_0))]^\top \nabla^2 U(q_0 + tp_0 -  s \frac{1}{2}t^2 \nabla U(q_0)) \mathrm{d}s  \right]^\top p_0 + \frac{1}{2}t [\nabla U(q_0 + tp_0) - \nabla U(\hat{q}_t)]^\top(\hat{p}_t-p_0)\\
&\leq \frac{1}{6}t^3 C_3 \|\mathsf{X}^\top p_0\|_\infty^2 \|\mathsf{X}^\top p_0\|_2 + \frac{1}{6}t^3 C_3 \|\mathsf{X}^\top p_0\|_\infty^2  \|\hat{p}_t - p_0\|_2\\
&\qquad \qquad  + \frac{1}{4}t^3 \int_0^1 p_0^\top \nabla^2 U(q_0) \times \nabla U(q_0) \mathrm{d}s\\
&\qquad \qquad+ \frac{1}{4}t^3 \int_0^1 p_0^\top \left[\nabla^2 U(q_0 + tp_0 -  s \frac{1}{2}t^2 \nabla U(q_0)) - \nabla^2 U(q_0)\right] \times \nabla U(q_0) \mathrm{d}s\\
&\qquad  \qquad+ \frac{1}{2}t [\nabla U(q_0 + tp_0) - \nabla U(\hat{q}_t)]^\top(\hat{p}_t-p_0)\\
&\leq \frac{1}{6}t^3 C_3 \|\mathsf{X}^\top p_0\|_\infty^2 \|p_0\|_2 + \frac{1}{6}t^3 C_3 \|\mathsf{X}^\top p_0\|_\infty^2  \|\hat{p}_t - p_0\|_2\\
&\qquad \qquad  + \frac{1}{4}t^3 M \|\nabla U(q_0)\|_2 |g|\\
&\qquad \qquad+ \frac{1}{4}t^3  C_3 \|\mathsf{X}^\top p_0\|_\infty  \left(\|\mathsf{X}^\top tp_0\|_{\infty} + \|\frac{1}{2}t^2 \nabla U(q_0)\|_2\right) \|\nabla U(q_0)\|_2\\
&\qquad  \qquad+ \frac{1}{2}t \times \|\frac{1}{2}t^2 \nabla U(q_0)\|_2 M \times \|\hat{p}_t-p_0\|_2\\
&\stackrel{{\scriptsize \textrm{Eq. }}\ref{eq:momentum_diff}}{\leq} \frac{1}{6}t^3 C_3 \|\mathsf{X}^\top p_0\|_\infty^2 \|p_0\|_2 + \frac{1}{6}t^3 C_3 \|\mathsf{X}^\top p_0\|_\infty^2 \times (2t M \|q_0\|_2 + \frac{1}{2}t^2 M \|p_0\|_2)\\
&\qquad \qquad  + \frac{1}{4}t^3 M^2 \|q_0\| \times |g|\\
&\qquad \qquad+ \frac{1}{4}t^3  C_3 \|\mathsf{X}^\top p_0\|_\infty  \left(\|\mathsf{X}^\top tp_0\|_{\infty} + \frac{1}{2}t^2  M \|q_0\|_2\right) M \|q_0\|_2 \\
&\qquad  \qquad+ \frac{1}{2}t \times \frac{1}{2}t^2 \|q_0\|_2 M^2 \times (2t M \|q_0\|_2 + \frac{1}{2}t^2 M \|p_0\|_2)\\
&\leq \frac{1}{100},
\ee
with probability at least $1-\frac{\epsilon}{\mathcal{I}}$,  whenever  $(q_0, p_0) \in G$,  where $g\sim N(0,1)$.  The last inequality holds because of our choice of $\eta$ and by the Hanson-Wright inequality.

\end{enumerate}

\paragraph{Combining terms.}

We now combine our bounds for the individual terms to bound the error in the Kinetic energy:

\be
\mathcal{K}(p_t) - \mathcal{K}(\hat{p}_t) \leq  &\frac{1}{6} t^3 C_3 \|\mathsf{X}^\top p_0\|_{\infty}^2 \|p_0\|_2 + \frac{1}{8} t^4 C_4  \|\mathsf{X}^\top p_0\|_{\infty}^4\\
&+  \frac{1}{6} t^3 C_3 \|\mathsf{X}^\top p_0\|_{\infty}^2 \times \left( t M \|q_0\|_2  + \frac{1}{2}t^2 M \|p_0 \|_2\right)\\
&+ \frac{1}{6} C_3 t^4 M \sup_{s\in[0,\tau]}\|q_s\|_2 \times \|\mathsf{X}^\top p_0\|_{\infty}^2 \left(  \|\mathsf{X}^\top p_0\|_\infty + t \|q_0 \|_2 M  + \frac{1}{2}t^2  M\|p_0\|_2 \right)\\
&+ \frac{5}{2}\left(\frac{1}{3} t^3 \left[C_3  \sup_{t\in[0,\eta]}\|\mathsf{X}^\top p_t\|_{\infty}^{2} +  (M)^2 \sup_{t\in[0,\eta]} \|q_t\|_2 \right]\right)^2\\
&+ \frac{1}{3} t^3 \left[C_3 \sup_{t\in[0,\eta]}\|\mathsf{X}^\top p_t\|_{\infty}^{2} +  (M)^2 \sup_{t\in[0,\eta]} \|q_t\|_2 \right] \times  \left[\|p_0\|_2 + t \|q_0 \|_2 M  + \frac{1}{2}t^2  M\|p_0\|_2 \right]\\
+\frac{1}{100}\\
&\leq \frac{2}{100}.
\ee
with probability at least $1-\frac{\epsilon}{\mathcal{I}}$,  whenever  $(q_0, p_0) \in G$.
\end{proof}

\section{Proof of main theorem for optimization}

\begin{proof}[Proof of theorem \ref{thm:main_optimization}]

Without loss of generality, we may assume that $U$ has a global minimizer $x^{\star}$ at $x^{\star} = 0$ (see comment at the beginning of the proof of theorem \ref{thm:main}).

We define the following lazy Markov chain $\hat{X}$:

Let $V_1, V_2 \ldots \sim N(0,I_d)$, and let $\hat{X}_{0} = X_0$.  For every $i$, let
\be
\hat{X}_{i+1} &= X_i + \eta V_i - \frac{1}{2}\eta^2 \nabla U(\tilde{X}_i)\\
\hat{V}_{i+1} &= V_i - \eta \nabla U(X_i)  - \frac{1}{2}\eta^2 \frac{\nabla U(\hat{X}_{i+1} )  - \nabla U(X_i)}{\eta}\\
Z_{i+1} &= \begin{cases}\hat{X}_{i+1} \qquad \textrm{ with probability } \min(1, \, e^{\mathcal{H}(\hat{X}_i,\hat{V}_i)-\mathcal{H}(X_i, V_i)})
\\ X_i \qquad \textrm{ otherwise}  \end{cases}\\
\tilde{Z}_{i+1} &= \begin{cases}Z_{i+1} \qquad \textrm{ if } Z_{i+1} \in \mathsf{S}
&\\ X_i \qquad \textrm{ otherwise}  \end{cases}\\
\tilde{X}_{i+1} &= \begin{cases}\tilde{Z}_{i+1} \qquad \textrm{ with probability } \frac{1}{2}
&\\ X_i.  \qquad \textrm{ otherwise.}  \end{cases}
\ee

Note that this Markov chain is lazy and satisfies the detailed balance equations for its stationary distribution $\pi(x) \propto e^{-U(x)\mathbbm{1}_{\mathsf{S}}(x)}$.

Set $\mathcal{I} =  \frac{4 \log(\frac{\beta}{\delta})}{(\Delta \hat{\Psi}_\pi(\mathsf{S}\backslash \mathcal{U}))^2}$.

Choose $\alpha = \log(\frac{\mathcal{I}\beta N}{\delta})$ and $R = \sqrt{d} \log(\frac{1}{\delta}\max(1,\frac{\sqrt{M}}{\mathsf{a} \mathcal{I} \beta N}))$, where $N = \lceil R M^\frac{1}{2} \eta \rceil$.

By Lemma \ref{lemma:WarmPreserved}, we have that
\be
\mathbb P((\tilde{X}_i, V_i) \in G \,\,  \forall \,\, 0 \leq i \leq \mathcal{I}-1) \geq 1- \mathcal{I} \times \beta \times [N r e^{-\frac{16 \alpha^2 -1}{8}} -e^{-\frac{R^2-d}{8}} - e^{-\frac{\mathsf{a}}{\sqrt{d}} \frac{R}{\sqrt{M}}}] \geq 1- \delta
\ee
Therefore, by Lemmas \ref{lemma:potential} and \ref{lemma:kinetic} with probability at least $1-\frac{\delta}{\mathcal{I}}$, the acceptance probability is
\be
\min(1, \, e^{\mathcal{H}(\hat{X}_i,\hat{V}_i)-\mathcal{H}(X_i, V_i)}) \geq e^{-10} > 0.99.
\ee

Let $i^\star = \min\{i : (\tilde{X}_i,V_i)\notin G \}$.  Then with probability at least $1- \mathcal{I}\times \frac{\delta}{\mathcal{I}} = 1-\delta$, we have that $\mathcal{I}\leq i^\star$.
 Consider the toy Markov chain $\tilde{X}^{\dagger}$, where
\be
\tilde{X}^{\dagger}_i = \begin{cases} \tilde{X}_i \qquad \textrm{ if } i < i^\star \\ Y_i \qquad \textrm{ if } i \geq i^\star, \end{cases}
\ee
where $Y_1, Y_2 \ldots \sim \pi$ are i.i.d. and each $Y_i$ is independent of $\tilde{X}_0,\ldots, \tilde{X}_{i-1}$.  Denote the transition kernel of $\tilde{X}^{\dagger}$ by $\tilde{K}^{\dagger}$.

Then by Lemma \ref{lemma:restricted_conductance2} we have that the restricted conductance $\hat{\Psi}_{\tilde{K}^\dagger}(\mathsf{S} \backslash [\mathcal{U}_{\Delta}]))$ of the $\tilde{X}^{\dagger}$ chain is $ \Omega(\Delta \hat{\psi}_{\pi}(\mathsf{S} \backslash \mathcal{U}))$, where $\Delta = \frac{1}{100}(\frac{1}{2}\eta^{-1} + \frac{1}{4} \eta M)^{-1}$.

Thus, by Lemma \ref{lemma:hitting}, we have:
\be
\inf\{i :\tilde{X}_i^{\dagger} \in \mathcal{U}_\Delta\} \leq \frac{4 \log(\frac{\beta}{\delta})}{\hat{\Psi}_{\textrm{$\tilde{K}^{\dagger}$}}^2(\mathsf{S}\backslash [\mathcal{U}_\Delta])}.
\ee
Hence,
\be
\inf\{i :\tilde{X}_i^{\dagger} \in \mathcal{U}_\Delta\} \leq \frac{4 \log(\frac{\beta}{\delta})}{\Delta^2 \hat{\psi}_{\pi}^2(\mathsf{S} \backslash \mathcal{U})}.
\ee

Therefore, since with probability at least $1-\delta$ we have $\tilde{X}_i = X^{\dagger}_i$, it must be that
\be \label{eq:m1}
\inf\{i :\tilde{X}_i \in \mathcal{U}_\Delta\} \leq \frac{4 \log(\frac{\beta}{\delta})}{\Delta^2 \hat{\psi}_{\pi}^2(\mathsf{S} \backslash \mathcal{U})},
\ee
with probability at least $1-2\delta$.

Since $\tilde{X}$ is the lazy version of the Markov chain $X$, and both chains start at the same initial point, inequality \eqref{eq:m1} implies that
\be
\inf\{i :X_i \in \mathcal{U}_\Delta\} \leq \frac{4 \log(\frac{\beta}{\delta})}{\Delta^2 \hat{\psi}_{\pi}^2(\mathsf{S} \backslash \mathcal{U})},
\ee
with probability at least $1-2\delta$.

\end{proof}

\section{Proof of Theorem \ref{thm:logit}} \label{sec:C4}
\begin{proof}
The proof of this theorem for $C_3$ for general loss functions $\varphi$ is identical to the proof of Theorem 2 of \cite{mangoubi2018dimensionally}, which was stated for the special case where $\varphi$ is the logistic loss function.

To bound $C_4$, we note that
\be
| \nabla^4U(x)[u,u, u, u] | \leq \sum_{i=1}^r |F^{(4)}(\mathsf{X}_i^\top x)|\times |\mathsf{X}_i^\top u|^4 \leq  \sum_{i=1}^r 1 \times \|\mathsf{X}^\top u\|_\infty^4 = r \|\mathsf{X}^\top u\|_\infty^4.
\ee
\end{proof}

\section{Proof of Theorem \ref{thm:ZeroOneLoss}}
Without loss of generality, we may assume that $U$ has a global minimizer $x^{\star}$ at $x^{\star} = 0$ (see comment at the beginning of the proof of theorem \ref{thm:main}).

Let $B = \{x\in \mathbb{R}^d : \|x\|_2 \leq 1\}$ be the unit ball.
\begin{lemma} \label{AdditiveNoise}
Let $\nu>0$ and suppose that $\lambda \geq \frac{100 \sqrt{d} }{\nu \log(\nu)}$.  Then we have
\be
|\tilde{f}(x) - f(x)| \leq  2\nu \qquad \forall x \in B \backslash \frac{1}{2} B.
\ee
\end{lemma}
\begin{proof}
From Lemma 8 in \cite{zhang2017hitting}, we have that $F$ is $6$-Lipschitz on $\mathsf{S} = B \backslash \frac{1}{2} B$.  Let $z$ be a point uniformly distributed on the unit sphere $\partial B$.

Then for any unit vector $u$, we have $\mathbb{P}(|u^\top z| \leq \frac{\nu}{10\sqrt{d}}) \leq \nu $.  Moreover, since we chose $\lambda \geq \frac{100 \sqrt{d} }{\nu \log(\nu)}$, we have that $1- \varphi(\lambda s) \leq \nu$ whenever $s \geq \frac{\nu}{10\sqrt{d}}$. 

Therefore,
\be
\mathbb{E}[|\tilde{f}(z) - f(z)|] &= \frac{1}{r}\textstyle{\sum_{i=1}^r} \mathbb{E}[ \hat{\ell}(\lambda z ; (\mathcal{X}_i, \mathcal{Y}_i)) - \ell(z ; (\mathcal{X}_i, \mathcal{Y}_i))]\\
&\leq \mathbb{P}((|\mathcal{X}_i^\top z| \leq \frac{\nu}{10\sqrt{d}}) +  \nu\\
&\leq 2 \nu.
\ee
\end{proof}

Fix some $\alpha_0 \in [0,\frac{\pi}{4}]$. Let $S := B\backslash (\frac{1}{2}B)$ where $B$ is the unit ball, and let $\mathcal{U} :=  \left \{x \in S : \left \langle \frac{x}{\|x\|} , \theta^\star \right \rangle  \geq \cos(\alpha_0) \right\}$.

We restate Lemma 8 and 9 in \cite{zhang2017hitting} for convenience:

\begin{lemma} [Lemma 9 in \cite{zhang2017hitting}] \label{lemma:cheeger_ZeroOne}
There is a universal constant $c_1$ such that for inverse temperature $\mathcal{T}^{-1}\geq \frac{c_1 d^{\frac{3}{2}}}{ \mathfrak{q}_0\sin^2(\alpha_0)}$, the restricted Cheeger constant $\hat{\psi}_{\hat{\pi}}(S \backslash \mathcal{U})$ of $\hat{\pi} \propto e^{-\mathcal{T}^{-1} F(x)} \mathbbm{1}_S(x)$ is at least $\hat{\psi}_{\hat{\pi}}(S \backslash \mathcal{U}) \geq \frac{1}{3} d$.
\end{lemma}

\begin{lemma} [Lemma 8 in \cite{zhang2017hitting}] \label{lemma:8}
$F$ is $3$-Lipschitz on $\frac{5}{4}B \backslash \frac{1}{4} B$.

For any $\nu, \delta >0$ if the sample size $r$ satisfies $r \geq \frac{d}{\nu^2} \mathrm{polylog}(d, \frac{1}{\nu}, \frac{1}{\delta})$, then w.p. at least $1-\delta$ we have $\sup_{\mathbb{R}^d \backslash \{0\}} |f(x) - F(x)| \leq \nu$.
\end{lemma}

\begin{proof}
The proof follows directly from Lemma 8 in \cite{zhang2017hitting}, since $f(x) = f(\frac{x}{\|x\|})$ and $F(x) = F(\frac{x}{\|x\|})$ for all $x \in \mathbb{R}^d \backslash \{0\}$.
\end{proof}

We can now prove Theorem \ref{thm:ZeroOneLoss}:
\begin{proof} [Proof of Theorem \ref{thm:ZeroOneLoss}.]

\textbf{Bounding the derivatives of the objective function.}

First, we bound the derivatives of  $\mathsf{\tilde{f}}$. 
By the Hanson-wright inequality, there is a universal constant $\mathsf{c} \geq 1$ such that $|\mathcal{X}_i^\top \mathcal{X}_j| \leq \frac{\mathsf{c}}{\sqrt{d}} \log(\frac{r^2}{\delta})$ for every $i,j\in [r]$ with probability at least $1-\delta$ (for convenience, we will use the same universal constant throughout the proof).  Hence, with probability at least $1-\delta$, the incoherence $\Phi$ satisfies
 \be
 \Phi := \max_{i\in[r]}  \sum_{j=1}^r |\mathcal{X}_i^\top \mathcal{X}_j|  \leq \mathsf{c} \frac{r}{\sqrt{d}}.
 \ee

Thus, by Theorem \ref{thm:logit} we have that Assumption \ref{assumption:derivatives} is satisfied with constants  $C_3 = d^{\frac{3}{4}}\times \frac{1}{r} \mathcal{T}^{-1} \sqrt{r} \sqrt{ \Phi} \leq d^{\frac{3}{4}}\times \mathsf{c} \frac{\mathcal{T}^{-1}}{r} r d^{-\frac{1}{4}} = d^{\frac{1}{2}} \mathcal{T}^{-1}$, and $C_4 = d^{\frac{4}{4}} \times \frac{\mathcal{T}^{-1}}{r} r = d \mathcal{T}^{-1}$.

Moreover, we have  $ \nabla^2 U(x) \preccurlyeq  \frac{1}{r} \sum_{i=1}^r \mathcal{T}^{-1} d^{-\frac{2}{4}}\mathcal{X}_i \mathcal{X}_i^\top$ for all $x \in \mathbb{R}^d$.
 Hence, by the Matrix Chernoff inequality \cite{tropp2012user}, we have $\lambda_{\mathrm{max}}( \nabla^2 U(x)) \leq d^{-\frac{1}{2}}\lambda_{\mathrm{max}}( \frac{1}{r} \sum_{i=1}^r \mathcal{T}^{-1} \mathcal{X}_i \mathcal{X}_i^\top) \leq d^{-\frac{1}{2}} \log(\frac{\mathsf{c}}{\delta}) \frac{1}{d} \mathcal{T}^{-1} = \log(\frac{\mathsf{c}}{\delta}) \frac{1}{d^{\frac{3}{2}}} \mathcal{T}^{-1} $ for all $x \in \mathbb{R}^d$ with probability at least $1-\delta$.  Hence, we can set $M=\log(\frac{\mathsf{c}}{\delta}) \frac{1}{d^{\frac{3}{2}}} \mathcal{T}^{-1} = \log(\frac{\mathsf{c}}{\delta})  \frac{2 c_1}{  \mathfrak{q}_0\sin^2(\alpha_0)}$ with probability at least $1-\delta$.

\textbf{Bounding the magnitude of the gradient.}
  Since $F$ is continuous and uniformly bounded on $\mathsf{S}$, and $F(x) = F(\frac{x}{\|x\|})$ for all $x \neq 0$, we have that $F$ attains a global minimum $x^\star_F$ on $\mathsf{S}$.  Without loss of generality we may assume that $\|x^\star_F\|_2 = \frac{3}{4} d^{\frac{3}{4}}$, so that $B(x^\star_F, \frac{1}{4}d^{\frac{3}{4}}) \subseteq \mathsf{S}$.
  
   Suppose (towards a contradiction) that $\|\nabla U(z)\|_2 \geq 8 d^{\frac{3}{4}} M$ for some $z \in \mathsf{S}$ with probability at least $\delta$.  Then since any two points in $\mathsf{S}$ can be connected by a path in $\mathsf{S}$ of length less that $\mathrm{pi}\times d^\frac{3}{4}$, we would have that $\|\nabla U(z) - \nabla U(x)\|_2 \leq 4 d^{\frac{3}{4}} M$ for all $x\in \mathsf{S}$.
   
   Thus,  with probability at least $\delta$, there would exist a point $y^\star \in B(x^\star_F, \frac{1}{4}d^{\frac{3}{4}}) \subseteq \mathsf{S}$ such that $U(y^\star) \leq U(x^\star_F)-  4 d^{\frac{3}{4}} M \times  \frac{1}{4}d^{\frac{3}{4}} = U(x^\star_F)-   d^{\frac{3}{2}} M = U(x^\star_F)-   \log(\frac{\mathsf{c}}{\delta}) \mathcal{T}^{-1} \leq U(x^\star_F)- \mathcal{T}^{-1} 10\nu$.
   
   But by Lemmas \ref{AdditiveNoise} and \ref{lemma:8}, with probability at least $1-\delta$ we have $|\mathcal{T}^{-1} F(x) -U(x)| \leq \mathcal{T}^{-1} 3\nu $, which is a contradiction.
   
   Hence, by contradiction we have that  
   \be \label{eq:GradientMagnitude}
   \|\nabla U(z)\|_2 &< 8 d^{\frac{3}{4}} M\\
   &= 8 d^{\frac{3}{4}}  \log(\frac{\mathsf{c}}{\delta})  \frac{2 c_1}{  \mathfrak{q}_0\sin^2(\alpha_0)},
   \ee
    for all $z \in \mathsf{S}$ with probability at least $1-\delta$.
   
\textbf{Bounding the probability of proposal falling outside $\mathsf{S}$.}

Let $z\in \mathsf{S}$ be the current point in the Markov chain, and let $\gamma_z := z+ \eta v - \frac{1}{2} \eta^2 \nabla U(x)$ where $v \sim N(0,I_d)$ be the proposed step.  Without loss of generality, we may assume that our coordinate basis is such that $\frac{z}{\|z\|_2} = e_1$ and $z[1]>0$ (otherwise we can just rotate the coordinate axis about the origin, and apply the same rotation to the argument of the potential function $U$).  First, we note that
\be
\mathsf{S} &= \left\{x \in \mathbb{R}^d :   \frac{1}{2} d^{\frac{3}{4}} \leq \sqrt{\sum_{i=1}^d x[i]^2} \leq d^{\frac{3}{4}} \right\}\\
& = \left\{x \in \mathbb{R}^d :\frac{1}{4}d^{\frac{6}{4}} - \sum_{i=2}^d x[i]^2 \leq x^2[1] \leq d^{\frac{6}{4}} - \sum_{i=2}^d x[i]^2 \right\}.
\ee

Without loss of generality, we may assume that $z[1] \geq 0$ (otherwise, we can rotate the coordinate basis to make  $z[1] \geq 0$).  

\textbf{Case 1:} First, consider the case where $z[1] \geq \frac{3}{4}$. 

Let $E_0$ be the event that $\|v\|_2^2 \leq d \log(\frac{\mathsf{c}}{d})$ and let $E_1$ be the event that $-1 \leq v[1] \leq - \frac{1}{3}$ and  $\|v\|_2^2 \leq d \log(\frac{\mathsf{c}}{d})$.  Then $\mathbb{P}(E_1 \cap E_0) \geq \frac{1}{10}$.

We have
\be
\gamma_z[1] &:= z[1]+ \eta v[1] - \frac{1}{2} \eta^2 \nabla U(x)^\top e_1\\
&\stackrel{{\scriptsize \textrm{Eq.} \ref{eq:GradientMagnitude}}}{\leq}  d^{\frac{3}{4}} + \eta v[1] + \frac{1}{2} \eta^2 \times 8 d^{\frac{3}{4}} M\\
&\leq  d^{\frac{3}{4}} - \frac{1}{3}\eta  + \frac{1}{2} \eta^2 \times 8 d^{\frac{3}{4}} M\\
&\leq  d^{\frac{3}{4}} - \frac{1}{3}\eta  + \frac{1}{2} \eta^2 \times 8 d^{\frac{3}{4}}  \log(\frac{\mathsf{c}}{\delta})  \frac{2 c_1}{  \mathfrak{q}_0\sin^2(\alpha_0)}\\
&\leq  d^{\frac{3}{4}} - \frac{1}{6}\eta,
\ee
if we choose $\eta \leq [\frac{1}{2} \times 8 d^{\frac{3}{4}}  \log(\frac{\mathsf{c}}{\delta})  \frac{2 c_1}{  \mathfrak{q}_0\sin^2(\alpha_0)}]^{-1}$.

Hence,
\be \label{eq:FirstCoordinate}
(\gamma_z[1])^2 \leq  d^{\frac{6}{4}} - \frac{1}{3}  d^{\frac{6}{4}} \eta + \frac{1}{36}\eta^2.
\ee

But if $E_0$ occurs we also have,
\be \label{eq:OtherCoordinates}
d^{\frac{6}{4}} - \sum_{i=2}^d \gamma_z[i]^2 &\geq d^{\frac{6}{4}} - \|\gamma_z-z\|_2^2\\
&\geq d^{\frac{6}{4}} - \|\eta v - \frac{1}{2} \eta^2 \nabla U(x)\|_2^2\\
&\geq d^{\frac{6}{4}} - \eta^2 \|v\|_2^2 - \frac{1}{4} \eta^4 \|\nabla U(x)\|^2\\
&\geq d^{\frac{6}{4}} - \eta^2 \|v\|_2^2 - \frac{1}{4} \eta^4 [8 d^{\frac{3}{4}}  \log(\frac{\mathsf{c}}{\delta})  \frac{2 c_1}{  \mathfrak{q}_0\sin^2(\alpha_0)}]^2\\
&\geq d^{\frac{6}{4}} - \eta^2 d \log(\frac{\mathsf{c}}{d}) - \frac{1}{4} \eta^4 [8 d^{\frac{3}{4}}  \log(\frac{\mathsf{c}}{\delta})  \frac{2 c_1}{  \mathfrak{q}_0\sin^2(\alpha_0)}]^2.
\ee

Therefore, inequalities \eqref{eq:FirstCoordinate} and \eqref{eq:OtherCoordinates}, together with our choice of $\eta$, imply that  
\be \label{eq:case1_upper}
(\gamma_z[1])^2 \leq d^{\frac{6}{4}} - \sum_{i=2}^d \gamma_z[i]^2,
\ee
if the event $E_1$ occurs.

We now show a lower bound:
\be \label{eq:case1_lower}
\gamma_z[1] &:= z[1]+ \eta v[1] - \frac{1}{2} \eta^2 \nabla U(x)^\top e_1\\
&\stackrel{{\scriptsize \textrm{Eq.} \ref{eq:GradientMagnitude}}}{\geq}  d^{\frac{3}{4}} + \eta v[1] - \frac{1}{2} \eta^2 \times 8 d^{\frac{3}{4}} M\\
&\geq   \frac{3}{4}  d^{\frac{3}{4}} - \eta  - \frac{1}{2} \eta^2 \times 8 d^{\frac{3}{4}} M\\
&\geq  \frac{3}{4}  d^{\frac{3}{4}} - \eta  - \frac{1}{2} \eta^2 \times 8 d^{\frac{3}{4}}  \log(\frac{\mathsf{c}}{\delta})  \frac{2 c_1}{  \mathfrak{q}_0\sin^2(\alpha_0)}\\
&\geq   \frac{3}{4} d^{\frac{3}{4}} - \frac{3}{6}\eta\\
&\geq \frac{1}{2}d^{\frac{3}{4}}\\
&\geq \sqrt{\frac{1}{2}d^{\frac{3}{4}}- \sum_{i=2}^d \gamma_z[i]^2},
\ee
where the second-to-last inequality holds because of our choice of $\eta$.

Therefore, Inequalities \ref{eq:case1_upper} and \ref{eq:case1_lower} together imply that
\be  \label{eq:Case1}
\gamma_z \in \mathsf{S} \quad \textrm{ if the event $E_1 \cap E_0$ occurs and $z[1] \geq \frac{3}{4}d^{\frac{3}{4}}$}.
\ee

\textbf{Case 2:}  Now consider the case where $\frac{1}{2} d^{\frac{3}{4}}\leq z[1] \leq \frac{3}{4}  d^{\frac{3}{4}}$.  The proof for this case is similar to the proof for case $1$:

Let $E_2$ be the event that $\frac{1}{3} \leq v[1] \leq 1$, and recall that  $E_0$ is the event that $\|v\|_2^2 \leq d \log(\frac{\mathsf{c}}{d})$.  Then  $\mathbb{P}(E_2) =  \mathbb{P}(E_1 \cap E_0) \geq \frac{1}{10}$.

We have
\be
\gamma_z[1] &:= z[1]+ \eta v[1] - \frac{1}{2} \eta^2 \nabla U(x)^\top e_1\\
&\stackrel{{\scriptsize \textrm{Eq.} \ref{eq:GradientMagnitude}}}{\geq}  \frac{1}{2}d^{\frac{3}{4}} + \eta v[1] - \frac{1}{2} \eta^2 \times 8 d^{\frac{3}{4}} M\\
&\geq \frac{1}{2} d^{\frac{3}{4}} + \frac{1}{3}\eta  - \frac{1}{2} \eta^2 \times 8 d^{\frac{3}{4}} M\\
&\geq \frac{1}{2} d^{\frac{3}{4}} + \frac{1}{3}\eta  - \frac{1}{2} \eta^2 \times 8 d^{\frac{3}{4}}  \log(\frac{\mathsf{c}}{\delta})  \frac{2 c_1}{  \mathfrak{q}_0\sin^2(\alpha_0)}\\
&\geq \frac{1}{2} d^{\frac{3}{4}} + \frac{1}{6}\eta.
\ee
if we choose $\eta \leq [\frac{1}{2} \times 8 d^{\frac{3}{4}}  \log(\frac{\mathsf{c}}{\delta})  \frac{2 c_1}{  \mathfrak{q}_0\sin^2(\alpha_0)}]^{-1}$.

Hence, we have
\be \label{eq:case2_lower}
\frac{1}{4}d^{\frac{6}{4}} - \sum_{i=2}^d \gamma_z[i]^2 &\leq  \frac{1}{4}d^{\frac{6}{4}} \leq (\gamma_z[1])^2.
\ee

We also have that
\be \label{eq:FirstCoordinate2}
\gamma_z[1] &:= z[1]+ \eta v[1] - \frac{1}{2} \eta^2 \nabla U(x)^\top e_1\\
&\stackrel{{\scriptsize \textrm{Eq.} \ref{eq:GradientMagnitude}}}{\leq}  \frac{1}{2}d^{\frac{3}{4}} + \eta v[1] + \frac{1}{2} \eta^2 \times 8 d^{\frac{3}{4}} M\\
&\leq \frac{1}{2} d^{\frac{3}{4}} + \eta  + \frac{1}{2} \eta^2 \times 8 d^{\frac{3}{4}} M\\
&\leq \frac{1}{2} d^{\frac{3}{4}} + \frac{1}{3}\eta  + \frac{1}{2} \eta^2 \times 8 d^{\frac{3}{4}}  \log(\frac{\mathsf{c}}{\delta})  \frac{2 c_1}{  \mathfrak{q}_0\sin^2(\alpha_0)}\\
&\leq \frac{7}{8} d^{\frac{3}{4}},
\ee
where the last inequality holds because of our choice of $\eta$.

But if $E_0$ occurs we have from Inequality \eqref{eq:OtherCoordinates} that,
\be \label{eq:OtherCoordinates2}
\sum_{i=2}^d \gamma_z[i]^2 \leq \eta^2 d \log(\frac{\mathsf{c}}{d}) + \frac{1}{4} \eta^4 [8 d^{\frac{3}{4}}  \log(\frac{\mathsf{c}}{\delta})  \frac{2 c_1}{  \mathfrak{q}_0\sin^2(\alpha_0)}]^2\
&\leq  \frac{1}{100} d^{\frac{3}{4}}.
\ee

Therefore, by Inequalities \eqref{eq:FirstCoordinate2} and \eqref{eq:OtherCoordinates2} we have that
\be \label{eq:case2_upper}
\|\gamma_z\|_2^2 \leq \frac{7}{8} d^{\frac{3}{4}} + \frac{1}{100} d^{\frac{3}{4}} \leq d^{\frac{3}{4}}.
\ee

Therefore, Inequalities \ref{eq:case2_upper} and \ref{eq:case2_lower} together imply that 
\be \label{eq:Case2}
\gamma_z \in \mathsf{S}
\ee
 if the event $E_1 \cap E_0$ occurs and $\frac{1}{2}d^{\frac{3}{4}} \leq z[1] \leq \frac{3}{4}d^{\frac{3}{4}}$.

Therefore, from Equations \eqref{eq:Case1}  and \eqref{eq:Case2}, we have that $\gamma_z \in \mathsf{S}$ with probability at least $\frac{1}{10}$ whenever $z\in \mathsf{S}$.

\textbf{Bounding the hitting time.}
Let $X= X_0, X_1,\ldots $ be the Markov chain generated by Algorithm \ref{alg:optimization}.  Let $\mathcal{U} :=  \left \{x \in S : \left \langle \frac{x}{\|x\|} , \theta^\star \right \rangle  \geq \cos(\alpha_0) \right\}$, where $\alpha_0 = \epsilon$.

Choose $\eta \leq \tilde{O}\left(\min \left(C_3^{-\frac{1}{3}}d^{-\frac{1}{6}} , d^{-\frac{1}{3}}, C_4^{-\frac{1}{4}}\right) \min(1, M^{-\frac{1}{2}})\right)$.  Let $\pi_2 \propto e^{-U}\mathbbm{1}_{\mathsf{S}}$. Then by Theorem \ref{thm:main_optimization} we have
\be
\inf\{i :X_i \in \mathcal{U}_\Delta\} \leq \mathcal{I},
\ee
with probability at least $1- \delta$, where $\mathcal{I} = \frac{4 \log(\frac{\beta}{\delta})}{\Delta^2 \hat{\psi}_{\pi_2}^2(\mathsf{S} \backslash \mathcal{U})}$ and $\Delta = \frac{1}{100}(\frac{1}{2}\eta^{-1} + \frac{1}{4} \eta M)^{-1}$.

But by Lemma \ref{lemma:cheeger_ZeroOne} we have  $\hat{\psi}_{\pi_1}(\mathsf{S} \backslash \mathcal{U}) \geq \frac{1}{3} d \times \frac{1}{\mathcal{T}^{\frac{1}{2}} \times d^{\frac{1}{4}} \lambda} = \frac{1}{300}d^{\frac{1}{4}} \nu \log(\nu)$, where $\pi_1(x) \propto e^{-\mathcal{T}^{-1} F(x)} \mathbbm{1}_{\mathsf{S}}$.  Therefore, by Lemmas \ref{lemma:8} and \ref{AdditiveNoise} we have $|\mathcal{T}^{-1} F(x) - U(x)| \leq 3 \mathcal{T}^{-1}\nu$ and hence that
\be
\hat{\psi}_{\pi_2}(\mathsf{S} \backslash \mathcal{U}) \geq  \hat{\psi}_{\pi_1}(\mathsf{S} \backslash \mathcal{U})  e^{-6 \mathcal{T}^{-1}\nu} \geq \frac{1}{300}d^{\frac{1}{4}}  \nu \log(\nu) e^{-6 \mathcal{T}^{-1}\nu}.
\ee
Choosing, $\nu = \mathcal{T}$ gives
\be
\hat{\psi}_{\pi_2}(\mathsf{S} \backslash \mathcal{U}) \geq \frac{1}{300}d^{\frac{1}{4}}  \nu \log(\nu).
\ee

For our choice of $\eta$ we have $\Delta = \frac{1}{100}(\frac{1}{2}\eta^{-1} + \frac{1}{4} \eta M)^{-1} = \Omega (\eta)$, and
\be
\eta = C_3^{-\frac{1}{3}} d^{-\frac{1}{6}} M^{-\frac{1}{2}} = d^{\frac{5}{12}} \mathcal{T}^{\frac{5}{6}} \log^{-\frac{1}{2}}(\frac{\mathsf{c}}{\delta})\\
=   \Theta \left(d^{\frac{5}{12}} \times  [\frac{\mathfrak{q}_0\sin^2(\alpha_0)}{d^{\frac{3}{2}}}]^{\frac{5}{6}} \times \log^{-\frac{1}{2}}(\frac{\mathsf{c}}{\delta})\right)\\
=   \Theta \left(d^{-\frac{10}{12}} \times  [\mathfrak{q}_0\sin^2(\alpha_0)]^{\frac{5}{6}} \times \log^{-\frac{1}{2}}(\frac{\mathsf{c}}{\delta})\right).
\ee

  Therefore,

\be
\mathcal{I} &= O\left(\frac{\log(\frac{\beta}{\delta})}{\eta^2 \hat{\psi}_{\pi_2}^2(\mathsf{S} \backslash \mathcal{U})} \right)\\
&= \tilde{O}\left(\frac{\log(\frac{\beta}{\delta})}{\eta^2 d^{\frac{1}{2}}  \nu^2} \right)\\
&= \tilde{O}\left(\frac{\log(\frac{\beta}{\delta})}{\eta^2 d^{\frac{1}{2}}  \mathcal{T}^2} \right)\\
&= \tilde{O}\left(d^{\frac{25}{6}} \mathfrak{q}_0^{\frac{11}{3}} \sin^{-\frac{22}{3}}(\alpha_0) \log(\frac{\mathsf{c}}{\delta})  \log(\frac{\beta}{\delta}) \right).\\
\ee
\end{proof}

\section{Simple bound for Random Walk Metropolis} \label{sec:RWM}

In this section we obtain a simple bound for the Random walk Metropolis algorithm.

\begin{algorithm}[H]
\caption{Random Walk Metropolis \label{alg:RWM}}
\textbf{input:} Zeroth-order oracle for $U: \mathbb{R}^d \rightarrow \mathbb{R}$, step size $\eta>0$\\
 \textbf{input:}   Initial point $Z_0 \in \mathbb{R}^d$.\\
 \textbf{output:} Markov chain $Z_0, Z_1, \ldots, Z_{i_{\max}}$ with stationary distribution $\pi \propto e^{-U}$.
\begin{algorithmic}[1]

\For{$i=0$ to $i_{\mathrm{max}}-1$}
Sample $V_i \sim N(0,I_d)$.\\
Set $\hat{Z}_{i+1} = X_i + \eta V_i$\\
Set \be
Z_{i+1} = \begin{cases}\hat{Z}_{i+1} \qquad &\textrm{ with probability } \min(1, \, e^{U(\hat{Z}_i)-U(Z_i)})
\\ X_i \qquad &  \textrm{ otherwise}  \end{cases}
\ee
\EndFor
\end{algorithmic}
\end{algorithm}

\begin{theorem} [RWM]\label{thm:RWM}
Suppose that $U$ has $M$-Lipschitz gradient on $\mathbb{R}^d$ and satisfies Assumption \ref{assumption:tails}. %
 Then given a $\beta$-warm start, for any step-size parameter $\eta \leq \tilde{O}(\frac{\mathsf{a}}{d M} \log^{-1}(\frac{4 \beta}{\epsilon}))$ there exists $\mathcal{I} = O( \eta^{-2} \psi_\pi^{-2} \log(\frac{\beta}{\epsilon}))$ for which $Z_i$ satisfies $\| \mathcal{L}(Z_i) - \pi\|_{\mathrm{TV}} \leq \epsilon$ for all $i\geq \mathcal{I}$.
\end{theorem}

\begin{proof}
Since $Z_0$ has a $\beta$-warm start, by Lemma \ref{lemma:WarmPreserved} and Assumption \ref{assumption:tails}, for every $s>0$ have
\be
\mathbb{P}\left(\sup_{i \leq \mathcal{I}} \|Z_i -x^\star\|_2> s \right) \leq \mathcal{I} \times \beta \times e^{-\frac{\mathsf{a}}{\sqrt{d}} s}.
\ee
Thus, setting $s = \frac{\sqrt{d}}{\mathsf{a}} \log(\frac{2\mathcal{I} \beta}{\epsilon})$ we have:
\be
\mathbb{P}\left(\sup_{i \leq \mathcal{I}} \|Z_i -x^\star\|_2> \frac{\sqrt{d}}{\mathsf{a}} \log(\frac{4\mathcal{I} \beta}{\epsilon}) \right) \leq \frac{1}{4}\epsilon.
\ee

Moreover, by the Hanson-wright inequality
\be
\mathbb{P}[\sup_{i \leq \mathcal{I}} \|V_i\|>\xi]  \leq \mathcal{I} e^{-\frac{\xi^2-d}{8}} \quad \quad  \textrm{ for } \xi> \sqrt{2d}.
\ee
Thus, setting $\xi = 5\sqrt{d} \log({\mathcal{I}}{\epsilon})$ we have:
\be
\mathbb{P}[\sup_{i \leq \mathcal{I}} \|V_i\|>5\sqrt{d}]  \leq \frac{1}{4}\epsilon.
\ee

Let $i^\star = \inf\{i \in \mathbb{Z}^\star : \|Z_i -x^\star \|_2> \frac{\sqrt{d}}{\mathsf{a}} \textrm{ or } \|V_i\|>5\sqrt{d} \}$.  Then
with probability at least $1-\frac{1}{2}\epsilon$, we have $i^\star > \mathcal{I}$.

Let $Y_0, Y_1,\ldots \sim \pi$  i.i.d. and independent of $Z_0, Z_1,\ldots$, and define the toy Markov chain $\tilde{Z}$ as follows:
\be
\tilde{Z}_i &= Z_i \qquad \forall i \leq i^\star\\
\tilde{Z}_i &= Y_i \qquad \forall i > \mathcal{I}.\\
\ee

Let $a_{z,v} :=  \min(1, \, e^{U(z + \eta v)-U(z)})$ be the acceptance probability for the toy chain from any $z\in \mathbb{R}^d$ with velocity $v$.  If $\|z -x^\star\|_2\leq \frac{\sqrt{d}}{\mathsf{a}}$ and $\|V_i\| \geq 5\sqrt{d}$, then
\be \label{eq:RWM1}
a_{z,v} = \min(1, \, e^{U(z+ \eta v)-U(z)}) &\geq  \exp \left(- \eta \|v\|_2 \times \sup_{x\in[z, z+\eta v]}   \| \nabla U(x)\|_2) \right)\\
&\geq \exp \left(- \eta \|v\|_2 \times \sup_{x\in[z, z+\eta v]} M \|x\|_2 \right)\\
&\geq \exp \left(- \eta \|v\|_2 \times M (\|z\|_2 + \eta \|v\|_2) \right)\\
&\geq \exp \left(- \eta \|v\|_2 \times M (\|z\|_2 + \eta \|v\|_2) \right)\\
&\geq \exp \left(- \eta 5\sqrt{d} \times M \left(\frac{\sqrt{d}}{\mathsf{a}} \log(\frac{4\mathcal{I} \beta}{\epsilon}) + \eta 5\sqrt{d}\right) \right)\\
&\geq \exp \left(- \eta 30 d \times M \frac{1}{\mathsf{a}} \log(\frac{4\mathcal{I} \beta}{\epsilon}) \right)\\
&\geq \frac{1}{3}.
\ee

Let $x,y\in \mathbb{R}^d$, and $v,w\sim N(0,I_d)$. Therefore, by Theorem 1.3 in \cite{devroye2018total}, we have
\be \label{eq:RWM2}
\|\mathcal{L}(x+ \eta v) - \mathcal{L}(y+ \eta v)\|_{\mathrm{TV}} &\leq \frac{\|x-y\|_2}{2 \eta}.
\ee
Let $K_{\mathrm{toy RWM}}$ be the transition kernel of $\tilde{Z}$.  Then by inequalities \ref{eq:RWM1} and \ref{eq:RWM2}, whenever $\|x-y\|_2 \leq \Delta$ for $\Delta = \eta$, we have
\be
\|K_{\mathrm{toy RWM}}(x,\cdot) - K_{\mathrm{toy RWM}}(y,\cdot)\|_{\mathrm{TV}} &\leq  1-\frac{1}{3}\times \frac{\|x-y\|_2}{2 \eta} \leq \frac{5}{6}.
\ee
Then by Lemma \ref{lemma:conductance} we have $\Psi_{K_{\mathrm{toy RWM}}} = \Omega(\Delta \psi_\pi)$.  Moreover, by Lemma \ref{lemma:conductance} there is an $\mathcal{I} = O( \Psi_{K_{\mathrm{toy RWM}}}^{-2} \log(\frac{\beta}{\epsilon}))$ such that
\be
\|\mathcal{L}(\tilde{Z}_i) - \pi\|_{\mathrm{TV}} \leq \frac{1}{2} \epsilon \qquad \forall i \geq \mathcal{I}.
\ee
But \be
\tilde{Z}_i = Z_i \qquad \forall i \leq i^\star,
\ee
 and $i^\star > \mathcal{I}$ with probability at least $1-\frac{1}{2}\epsilon$.  Therefore,
\be
\|\mathcal{L}(Z_i) - \pi\|_{\mathrm{TV}} \leq \frac{1}{2} \epsilon + \frac{1}{2} \epsilon = \epsilon \qquad \forall i \geq \mathcal{I},
\ee
where $\mathcal{I} = O( \Psi_{K_{\mathrm{toy RWM}}}^{-2} \log(\frac{\beta}{\epsilon})) = O( \eta^{-2} \psi_\pi^{-2} \log(\frac{\beta}{\epsilon}))$.
\end{proof}

\appendix

\section{Hanson-wright inequality}
In this Appendix we recall the Hanson-Wright inequality \cite{hanson1971bound}, for the special case of Gaussian random vectors.

\begin{lemma}[Hanson-Wright inequality]
Let $Z\sim N(0,I_d)$ be a standard Gaussian random vector.  Then
\be
\mathbb{P}[\|Z\|_2>\xi]  \leq e^{-\frac{\xi^2-d}{8}} \quad \quad  \textrm{ for } \xi> \sqrt{2d}.
\ee
\end{lemma}

\bibliographystyle{plain}
\bibliography{MALA}

\end{document}